\documentclass{article}
\usepackage[utf8]{inputenc}
\usepackage{amsmath,amssymb}

\usepackage{xspace}

\usepackage[round]{natbib}
\usepackage{amsmath,graphicx}
\usepackage{float}
\usepackage[subrefformat=parens,labelformat=parens]{subcaption}
\usepackage{fullpage}

\usepackage{enumitem}
\usepackage{tabularx}

\usepackage{xcolor}

\definecolor{customcolor}{HTML}{191970}

\usepackage[colorlinks=true,allcolors=customcolor]{hyperref}

\usepackage[nameinlink]{cleveref}

\usepackage{color}
\definecolor{emphcol}{gray}{.75}
\newcommand{\block}[1]{
	\begin{center}
		\colorbox{emphcol}{
			\fbox{
				\begin{minipage}[t]{0.92\textwidth}
					#1
				\end{minipage}
			}
		}
	\end{center}
}

\usepackage{amsthm}
\newtheorem{theorem}{Theorem}
\newtheorem{definition}{Definition}

\newtheorem{lemma}{Lemma}

\newtheorem{example}{Example}

\newtheorem{remark}{Remark}

\newtheorem{corollary}{Corollary}
\newtheorem{assumption}{Assumption}

\newcommand{\token}{\mathbb{TK}}
\newcommand{\price}{\mathbb{PRICE}}
\newcommand{\TM}[2]{#1^{(#2)}}
\newcommand{\tm}[1]{\TM{#1}{t}}

\newcommand{\disc}[1]{#1^{(\delta, \infty)}}
\newcommand{\discst}[2]{#1^{(\delta, \infty|#2)}}

\newcommand{\VE}{\mathcal{V}}
\newcommand{\UE}{\mathcal{U}}

\newcommand{\Rew}{\mathcal{R}}
\newcommand{\Ser}{\mathcal{S}}

\newcommand{\Roptsingle}{R_{single}}

\title{\bf Single-token vs Two-token Blockchain Tokenomics}

\author{
  Aggelos Kiayias\\
  University of Edinburgh, IOG\\
  \texttt{aggelos.kiayias@ed.ac.uk}
  \and
  Philip Lazos\\
  IOG\\
  \texttt{philip.lazos@iohk.io}
  \and
  Paolo Penna\\
  IOG\\
  \texttt{paolo.penna@iohk.io}
}

\date{\today}

\begin{document}
	
\maketitle
\begin{abstract}
	We study long-term equilibria that arise in the 
    token monetary policy, or {\em tokenomics},  design of  proof-of-stake (PoS) blockchain systems that engage utility maximizing \emph{users} and \emph{validators}.
 Validators are system maintainers who get rewarded with tokens for performing the work necessary for the system to function properly, while users compete and pay with such tokens for getting a  desired portion of the  system service. 
	
	We study how the system service provision and suitable rewards schemes together can lead to equilibria with the following desirable characteristics (1) viability: the system keeps parties engaged, 
 (2) decentralization and skin-in-the-game: multiple sufficiently invested validators are participating, 
 (3) stability: the price path of the underlying token used to transact with the system does not change widely over time,  and
 (4)  feasibility: the mechanism is easy to implement as a smart contract, e.g., it does not require a fiat reserve  on-chain to perform token {\em buybacks} or to perform bookkeeping of exponentially growing token holdings.
 
 Our analysis enables us to put forward a novel generic mechanism for blockchain monetary policy that we call {\em quantitative rewarding} (QR). We investigate how to implement QR in single-token and two-token proof of stake (PoS) blockchain systems. The latter are systems that utilize one token for the users to pay the transaction fees and a different token for the validators to participate in the PoS protocol and get rewarded. 
 Our approach demonstrates a concrete advantage  of the two-token setting in terms of the ability of the QR mechanism to be realized effectively and provide good equilibria. Our analysis also reveals an inherent limitation of the single token setting in terms of implementing an effective blockchain monetary policy --- a distinction that is, to the best of our knowledge, highlighted for the first time.

\end{abstract}
\maketitle



\section{Introduction}
Blockchains create value by offering services in a fully decentralized manner, wherein \emph{users} pay fees to access these services, while the functioning and  security of the system is guaranteed by a set of nodes or \emph{validators} who receive rewards for performing the necessary computations required by  the protocol. 
These payments are issued in the system's native \emph{token} 
and the mechanism that mints and distributes these tokens to the relevant participants  
determines the  ``tokenomics policy'' of the blockchain. Designing such policies with
good properties is pivotal to ensuring the success of blockchain systems in the long term.

The token's value or \emph{price}, denominated in standard (fiat) currency, 
crucially determines the actual costs for users and the compensation for validators. 
An upward or downward fluctuation in the token's price can make the system less attractive for either type of party. While the system cannot directly control its token price in the market, it can implement various monetary policies (such as increasing token minting, burning transaction fees, adjusting the level of transaction fees, change the validator level of rewards, and others) to achieve a long-term equilibrium with the desired price without compromising the system's viability and decentralization.

This motivates the study of tokenomics design achieving the following important desiderata:  
\begin{enumerate}
	\item {\bf Viability}. Preconditioned on a positive system value, the system  keeps all involved parties actively engaged: validators to guarantee the protocol being live and securely running, and the users to guarantee enough fees are collected at all time.
	\item {\bf Decentralization and ``Skin-in-the-game.''} A set of validators each equally engaged in the protocol,  receiving adequate rewards to offset their costs, while also having significant stake in the system operation (e.g. in the case of proof-of-stake (PoS) systems, validators ``staking'' high enough amount of tokens).   
	\item {\bf Stability}. Reasonably stable token prices, meaning that the price of the token required to issue transactions does not significantly change over time. 
	\item {\bf Feasibility}. The tokenomics policy  should be feasible to implement on-chain algorithmically. In particular, 
 this means that the policy avoids using features that are incompatible or difficult to implement in the form of smart 
 contracts, integrate into the blockchain accounting model or
 they are antithetical to blockchain decentralization. 
For example, the policy should minimize the use of off-chain
entities that need to intervene actively to adjust policy (e.g., a central bank type of actor), 
or the use of techniques such as buybacks 
which would require a fiat reserve to facilitate. Furthermore, all the numerical quantities that are accounted in the ledger (e.g., the number of tokens held by a validator) should {\em not} grow exponentially over time. 

 \end{enumerate}
The conditions under which there is a tokenomics policy achieving \emph{all} these requirements is a fundamental question which we address in this work.

\subsection{Main contributions}

We explore tokenomics policies in the PoS setting, i.e., the setting where validators have to acquire tokens and stake them in order to receive rewards and perform the necessary system maintenance to further  the system's operation. 
We put forth a 
model where a set of $m$ users and $n$ validators engage with the system in discrete time steps. 
At each time step, users and validators buy or sell tokens in order to accommodate their objectives that include issuing transactions and staking tokens to provide the service. 

We give both boundary conditions and actual policies (mechanisms) that implement equilibria with all desired  features: viability, decentralization, stability, and feasibility. Specifically, our analysis leads us to put forward a novel generic mechanism
for blockchain ``monetary policy'': 
\begin{quote}
    {\bf Quantitative Rewarding (QR).}
The amount of rewards the system gives to the validators is suitably adjusted at each step to anticipate changes in the value of the service. 
\end{quote}

This idea stems from our analysis of equilibria that maintain a given price path (e.g., stable prices). Specifically: 

\begin{itemize}
	\item \emph{Prices and equilibria.} We show that,  in all equilibria implementing certain desired price paths, at every  round  (i) users spend a constant fraction of the total value of the offered service, and (ii) validators stake a constant fraction of the total rewards offered by the system (cf. Theorem~\ref{tH:prices_implies_generic}).
\end{itemize}

 A key observation here is that these two quantities are not necessarily related to each other, and the system can actually adjust the rewards in order to match changes in the value of the service  
 and control the prices to a desired path without directly participating in the token market (e.g., as it would be the case via token buybacks). 
Observe that if, at any point in time, the value of the service decreases, users will buy fewer tokens, thereby \emph{reducing demand} and potentially necessitating token buybacks if one is to keep prices stable. Rather than engaging in token buybacks, the key idea behind our QR mechanism is, perhaps somewhat counterintuitively, to \emph{increase the rewards} allocated to validators. As validators compete for the rewards, they react by staking more tokens, which they must be acquired from the market. This action effectively restores the demand for tokens to align with the supply. The reverse effect takes place when the service value increases. 

\begin{itemize}
    \item {\em Single Token Equilibria}. We investigate how to implement QR  via a  mechanism which adjusts the rewards at every round based on the current rewards and the next round service value using a single token for accounting. 
    An important challenge that arises in this setting is that the rewards formula is subject to a recursive condition (see the lower bound in Theorem~\ref{th:gen-symm-equilibrium-no-buy-back} and Corollary~\ref{cor:reward-growth-single}) 
    which makes rewards subject to a ``feedback loop'' effect that can lead to rewards exponential explosion. 
\end{itemize}

As a result,   while maintaining certain (minimal) rewards to ensure decentralization and security, it is possible that the system can
trigger an uncontrolled growth in the amount  of rewards distributed and in the staked amounts of tokens (Section~\ref{sec:dilemma-rewards-growth}).   
In Section~\ref{sec:implementation-single-token} we describe in detail the implementation of our mechanism for the case of stable prices. This implementation includes the strategies for users and validators, showing that all information needed is service value at the next round  (despite their simplicity, these strategies still constitute an equilibrium in the underlining repeated game, i.e., when more complex time-dependent strategies are in principle possible). 

The above results apply to the single token setting which is the most common in PoS systems (e.g., Ethereum and Cardano operate in this way, for example, as PoS cryptocurrencies)
and paint a rather negative picture: one has to either accept tokenomics with widely fluctuating token prices (the most common outcome that is observed in these systems as deployed in the real world), implement drastic measures such as token buybacks, or suffer a potential explosion in rewards accounting (unless future shocks in service value can be predicted --- a rather unreasonable assumption). 

Poised to obtain a more favorable outcome, we next explore the implementation of QR in two token PoS systems
in which one token is used to get the service and another token is used for staking (and governance). 
While less common, this type of setup has been adopted in a  few occasions: e.g.,  in the Neo cryptocurrency\footnote{See \url{https://docs.neo.org/docs/en-us/index.html}.} and can be implemented in other systems as well (e.g., in the Cosmos SDK\footnote{See \url{https://docs.cosmos.network/main/learn/beginner/gas-fees}.}). 
To investigate the potential of QR in this setting, we adapt  our model to the two tokens setting.
In order to maintain feasibility, 
while maintaining desirable prices, we identify the necessary conditions that equilibria (mechanism) must guarantee in terms of rewards:
\begin{itemize}
    \item \emph{Double Token Equilibria.} The two token mechanism
    for QR is presented in Section~\ref{sec:stable-price-mechanism-two}. The main advantage of the resulting rewards formula is that it is  simpler than in the single token setting and the mechanism does not need to know ``global information'' about the service value time series --- merely a one step lookahead suffices  (cf.   Theorem~\ref{th:generic-nobuy-back-tow}). As an important corollary, the mechanism achieves our objectives and is capable of avoiding the feedback loop effect of the single token case that leads to an exponential explosion for rewards. 
\end{itemize}

Intuitively, the advantage of decoupling the users from the validators, using different tokens, allows to better react and absorb shocks and fluctuations in the value of the service  while maintaining the price of the users' token stable and keeping all our feasibility desiderata. The price of the token used for staking increases gradually, while the amount of staked tokens remains constant, without scarifying incentive compatibility. 
This should be compared and contrasted with the single-token setting, which  
may enter a feedback loop and require an increase of the  validator rewards indefinitely along an exponential curve (even e.g.,  after  the service time series stabilizes) --- something that appears to be an inherent distinction between the single token and the two token setting in our model. 

\medskip 

\noindent {\bf Related work.} In the single-token setting, our results build on  the model of \cite{decentralization_friction}. Compared to that prior work, our results  accommodate a more general class of equilibria (see Appendix~\ref{app:eq:simpler} for details) with respect to  their features, that facilitate 
stable  prices without the requirement to employ  buybacks. 
 To the best of our knowledge, ours is the first analytical model for such systems to study their equilibria and related token price paths 
 achieving all four desiderata. Prior work includes the effects of tokens regarding user adoption and equilibrium selection \cite{bakos2019role, bakos2022overcoming,  li2018digital,sockin2023model}, or using token supply and buybacks to promote optimal growth and service provision~\cite{cong2021tokenomics, cong2022token}. Different aspects of the monetary characteristics of tokens have been studied in  \cite{schilling2019some, pagnotta2022decentralizing, prat2021fundamental}, while \cite{kiayias2023would} considered token burning; notably one can view QR as a generalization of the token burning mechanism that is more versatile, cf. \ref{remark:burning}. In prior work, typically users receive the most attention, but \cite{chitra2021competitive} expand on the strategy space of validators and their alternative uses for tokens. Further works modeling validators staking strategies are \cite{cong2022tokenomics,jermann2023macro,saleh2021blockchain,john2021equilibrium}. While all previous work has focused on the single token design, \cite{Dimitri2023}  considered two-token systems and proposed several economic indicators.

\section{Model description and notation (single token)}\label{sec:model}
\begin{figure}

\block{
\paragraph{Notation (single token model)}

\begin{itemize}
    \item $m$ = number of users (fixed and constant over time).
    \item $n$ = number of validators (fixed and constant over time).
    
    \item $\token_p$ = token holding of a generic player $p$ (user or validator).
    \item $\price$ = price of the token. 
    \item $\tm{f}$ the value of a generic quantity $f$ at time $t$.
    \item $\disc{f}$ = discounted version of quantity $\tm{f}$, that is, $\disc{f} = \sum_{t=0}^{\infty} \delta^t \tm{f}$. 
    \item $\discst{f}{t}$ = generalization of the previous definition in which we start at $t$ and discount accordingly, that is, $\discst{f}{t} = \sum_{\tau =t}^{\infty} \delta^{\tau - t} \TM{f}{\tau} = \TM{f}{t} + \delta\discst{f}{t+1}$.
    \item $\tm{s_j}$ = validator $j$'s selling strategy at time $t$ (negative values are possible, i.e., buying).
    \item $\tm{b_i}$ = user $i$'s buying strategy at time $t$ (negative values are possible, i.e., selling).
    \item $\tm{u_i}$ = amount of tokens that user $i$ pays for the service at time $t$. 
    \item $v$ = cost incurred by each validator (identical for all validators and all $t$).
    \item $\tm{R}$ is the total reward for validators at time $t$, which is further distributed to each validator according to a reward sharing scheme $r(\cdot)$, meaning that validator $j$ receives an amount of tokens equal to
    \begin{align}
        \label{eq:rss}
        \tm{R_j}  := \tm{R} \cdot  r(\tm{\token_j}, \tm{\token_{-j}}) \ . 
    \end{align}
    where $\tm{\token_{-j}}$ are the token holding of all validators but $j$.
    \item $\tm{S}$ is the service value at time $t$, which is further divided among the users according to some scheme $s(\cdot)$, meaning that each user $i$ receives a service value
    \begin{align}\label{eq:pay-scheme}
        \tm{S_i} := \tm{S} \cdot s(\tm{u_i}, \tm{u_{-i}}) \ . 
    \end{align}
    where $\tm{u_{-i}}$ are the used tokens  of all users but $i$.
    \item $g(n)$ = generic non-decreasing decentralization factor, where $n$ is the number of validators.
    \item $\tm{U_i}$ = instantaneous utility of user $i$, given by 
    \begin{align}\label{eq:usr-utility-istant}
        \tm{U_i} = \tm{S_i} \cdot g(n) - \tm{b_i} \cdot \tm{\price} \ . 
    \end{align}
    \item $\tm{V_j}$ = instantaneous utility of validator $j$, given by 
    \begin{align}\label{eq:val-utility-istant}
        \tm{V_j} = \tm{s_j} \cdot \tm{\price} - v \ . 
    \end{align}
\end{itemize}
}

  \caption{Notation and symbols.}
    \label{fig:notation}
\end{figure}
\noindent
We consider a Lagos-Wright \citep{lagos2005unified} type of model, in which there are two kinds of players: users and validators, as in \cite{kiayias2023would,decentralization_friction}. 
We have a repeated game where each round $t$,  consists of two phases or \emph{subrounds}. Players first interact with the system (user pay the system to get some service, and validators stake tokens to access rewards for performing some task),  and then they interact with a ``spot market'' (players buy or sell their tokens at the end of the round in order to hold the amount of tokens they need at the next round). In order to describe the model in detail, we introduce some notation in Figure~\ref{fig:notation}. The two subrounds are as follows:

\paragraph{Subround~I (pay \& stake):} Each user $i$ and each validator $j$ has $\tm{\token_i}$ and $\tm{\token_j}$ tokens, respectively.  Then,  given some (total) value of the service $\tm{S}$ and the total rewards $\tm{R}$ for validators, we have 
\begin{enumerate}
    \item Each validator $j$ performs some work, incurs a cost $v$, and receives an amount $\tm{R_j}$ of additional tokens according to \eqref{eq:rss}, given the total reward for validators $\tm{R}$.
    \item Each user $i$ uses some of her currently available tokens to pay for the service, and receives the corresponding service value $\tm{S}_i$ according to \eqref{eq:pay-scheme}, given the value of the service $\tm{S}$. 
\end{enumerate}
\paragraph{Subround~II (buy or sell):} Each user $i$ and each validator $j$ has $\tm{\token_i} - \tm{u_i}$ and $\tm{\token_j} + \tm{R_j}$ tokens, respectively.  Both type of players (users and validators) can buy any amount of new tokens or sell (part or all of the tokens they currently have) at the current market price. In detail:
\begin{enumerate}
    \item Each validator $j$ sells $\tm{s_j}$ tokens to the market and receives $\tm{s_j} \cdot \tm{\price}$ units of money in return. Note that we have a validator's selling constraint 
    \begin{align}\label{eq:val-sell-constraint}
       \tm{s_j} \leq \tm{\token_j} + \tm{R_j} \ . 
    \end{align}
    Any \emph{negative} $\tm{s_j}$ is allowed meaning that $j$ is actually \emph{buying} tokens from the market at the current price. 
    \item Each user $i$ buys $\tm{b_i}$ additional tokens from the market and pays $\tm{b_i} \cdot \tm{\price}$ units of money for that. Any \emph{negative} $\tm{b_i}$ is allowed, meaning that $i$ is \emph{selling} tokens to the market at the current price. Note that we have a users's selling constraint 
    \begin{align}\label{eq:usr-sell-constraint}
       -\tm{b_i} \leq \tm{\token_i} - \tm{u_i} \ . 
    \end{align}
\end{enumerate}
At the end of subround II of step $t$ we get the  token holdings for the next step, $t+1$, for each user $i$ and each validator $j$, respectively
\begin{align}\label{eq:token-holdings}
    \TM{\token_i}{t+1} = \tm{\token_i} - \tm{u_i} + \tm{b_i} \stackrel{\eqref{eq:usr-sell-constraint}}{\geq} 0 && \TM{\token_j}{t+1} = \tm{\token_j} + \tm{R_j} - \tm{s_j} \stackrel{\eqref{eq:val-sell-constraint}}{\geq} 0 \ . 
\end{align}
Each user $i$ and each validator $j$ aims at maximizing her own  \emph{discounted} utility, respectively 
\begin{align}
    \label{eq:usr-utility-discounted}
    \disc{U_i} & \stackrel{\eqref{eq:usr-utility-istant}}{=} \sum_{t=0}^\infty \delta^t \cdot [\tm{S_i} \cdot g(n) - \tm{b_i} \cdot \tm{\price}] \, 
    \\ 
    \label{eq:val-utility-discounted}
    \disc{V_j} & \stackrel{\eqref{eq:val-utility-istant}}{=} \sum_{t=0}^\infty \delta^t \cdot [\tm{s_j} \cdot \tm{\price} - v] 
\end{align}
subject to their respective selling constraints in \eqref{eq:usr-sell-constraint} and in  \eqref{eq:val-sell-constraint}. As usual, the discount factor $\delta$ above is $\delta=1/(1+r)$ where $r$ is the risk free rate.

\begin{definition}[symmetric equilibrium]\label{def:weakly-symmetric-equilibrium}
    Consider a system policy given by \eqref{eq:rss} and \eqref{eq:pay-scheme}, and a triple of users' strategies, validators' strategies, and  prices
    \begin{align*}
      (\tm{u_i}, \tm{b_i}) &&  \tm{s_j} &&  \tm{\price}
    \end{align*} 
    such that the corresponding selling constraints \eqref{eq:usr-sell-constraint} and \eqref{eq:val-sell-constraint} are satisfied. We say that such a triple is an \emph{equilibrium} if 
    \begin{enumerate}
        \item  Strategy $\tm{s_j}$ maximizes the discounted utility \eqref{eq:val-utility-discounted} of validator $j$,  given all other strategies, for each validator $j$;
        \item Strategy $(\tm{u_i}, \tm{b_i})$ maximizes the discounted utility \eqref{eq:usr-utility-discounted} of user $i$,  given all other strategies, for each user $i$;
        \item The resulting two sequences of token holdings \eqref{eq:token-holdings} are strictly positive, that is, $\tm{\token_i}>0$ and $\tm{\token_j}>0$.
    \end{enumerate}
    Moreover, we say that such an equilibrium is \emph{ symmetric} if the token holdings of all users are the same, and similarly, if all token holdings of all validators are the same,  correspond to sequences of strictly positive token holdings
    \begin{align}\label{eq:symmetric--token-holdins}
        \tm{\token_i}= \tm{\token_U}>0 && \tm{\token_j} = \tm{\token_V}>0 \ . 
    \end{align}
    for all users $i$ and for all validators $j$ and for all $t$.
\end{definition}

\begin{remark}[viability and equilibria]\label{rem:positive-token-holding}
The condition that tokens holding must be strictly positive at all time steps implies that the selling/buying strategies at such an equilibrium satisfy the selling constraints \eqref{eq:usr-sell-constraint} and \eqref{eq:val-sell-constraint}  with a strict inequality. This condition captures our \emph{viability} requirement of keeping all parties engaged. 
\end{remark}

In order to establish the existence of  equilibria, we need a few (technical) assumptions, regarding the service value (Assumption~\ref{asm:service-level}) and about the rewards and service sharing functions (Assumption~\ref{asm:schemes} below).   

\begin{assumption}[service value]\label{asm:service-level}
The service value $\tm{S}$ is upper bounded by some (arbitrarily large) constant independent of $t$.   
\end{assumption}
The value of the constant in Assumption~\ref{asm:service-level} does not affect the results and bounds in any way, and is only needed 
to establish the existence of equilibria. 
Other than this, we make no other assumptions, and allow $\tm{S}$ to increase or  decrease arbitrarily within this range.

\begin{remark}[service value  fluctuations -- shocks]\label{rem:shocks}
    The  service value $\tm{S}$ represents the overall value the system provides to users. Several factors (technological improvements, competitors, regulations, market conditions) influence $\tm{S}$. Consequently, the system may neither fully control nor predict all future values of $\tm{S}$, as these factors can cause unforeseen shocks. However, the system can attempt to counteract changes in $\tm{S}$ by adjusting other parameters, such as the rewards $\tm{R}$.    
\end{remark}
\section{Single token mechanisms}\label{sec:single-token}
\subsection{Analysis of price path}\label{sec:result}
In this section we provide the analysis of the  prices for the following family of service allocation and rewards:
\begin{align}\label{eq:general-schemes}
   \tm{S_i} = \tm{S}\cdot s\left(\frac{\tm{u_i}}{\sum_{k} \tm{u_k}}\right) && \tm{R_j} = \tm{R} \cdot r\left(\frac{\tm{\token_j}}{\sum_v \tm{\token_v}} \right) 
\end{align}
where $s(\cdot)$ and $r(\cdot)$ are arbitrary differentiable functions in $(0,1)$, and where the indexes $k$ and $v$ in the summations range over all users and all validators, respectively. In the following, we denote by $r'(\cdot)$ and $s'(\cdot)$  the first derivatives of $r(\cdot)$ and $s(\cdot)$, respectively. We make the following assumption about $r(\cdot)$ and $s(\cdot)$.

\begin{assumption}\label{asm:schemes}
    We assume that the functions $r(\cdot)$ and $s(\cdot)$ in \eqref{eq:general-schemes} are both concave.  Moreover, for the number $m\geq 2$ and $n\geq 2$ of users and validators under consideration, they never allocate more than the total rewards or service value available, $r(1/n)\leq 1/n$ and  $s(1/m)\leq 1/m$,  and the first derivatives satisfy $r'(1/n)>0$ and $s'(1/m)>0$. 
\end{assumption}

\begin{lemma}\label{le:prices-general}
    In any  symmetric equilibrium, it holds that
    \begin{align}
    \label{eq:val-price}
        \TM{\price}{t-1}  = \tm{\price} \cdot \delta \cdot  \tm{\Rew} \ ,   &&  \tm{\Rew} =1 + \frac{\tm{R}}{n\tm{\token_V}}\cdot \frac{n-1}{n} \cdot r' (\frac{1}{n})\ . 
    \end{align}
    Moreover 
    \begin{align}
    \label{eq:usr-price}
        \TM{\price}{t-1} = \delta \cdot \begin{cases}
        \tm{\Ser} & \text{if  } \tm{u_i} = \tm{\token_U} \\
        \tm{\price} & \text{if  } \tm{u_i} < \tm{\token_U} 
    \end{cases} \ ,
    && \tm{\Ser} = \frac{\tm{S}}{m\tm{\token_U}} \cdot g(n) \cdot \frac{m-1}{m}\cdot s'(\frac{1}{m}) \ . 
    \end{align}
    Hence, if $\tm{\Rew} \neq 1$, then 
$\tm{u_i} = \tm{\token_U}$ and $\tm{b_i}=\TM{\token_U}{t+1}$ for all users $i$, and the following identity must hold:
\begin{align}
\label{eq:usr-val-connection}
\tm{\price}\tm{\Rew} = \tm{\Ser} \ . 
\end{align}   
\end{lemma}

A simple class of schemes satisfying our assumptions is the following one:
\begin{align}\label{eq:simpler-schemes}
    \tm{S_i} = \tm{S} \cdot \left(\frac{\tm{u_i}}{\sum_{k} \tm{u_k}}\right) && \tm{R_j} = \tm{R} \cdot \left(\frac{\tm{\token_j}}{\sum_v \tm{\token_v}} \right) \ .
\end{align}

Example~\ref{ex:scheme-alternative} below provides another possible reward sharing function $r(\cdot)$ other than the one in \eqref{eq:simpler-schemes} -- see also Appendix~\ref{sec:special rewards} for further examples. We do not claim that this alternative $r(\cdot)$ provides any particular improvement, but simply point out that
changing $r(\cdot)$, and similarly $s(\cdot)$,  does affect the equilibria as described by   Lemma~\ref{le:prices-general} (see Appendix~\ref{app:eq:simpler} for details and further discussion).

\begin{example}\label{ex:scheme-alternative}
    For any parameter $\ell>1$, the following  reward sharing function $r(x) = x - x^\ell$ satisfies $r(1/n) = 1/n - 1/ n^\ell >0$, meaning that for smaller number of validators a smaller fraction of the total allocated rewards is actually distributed. Since $r'(1/n) = 1 - \ell/n^{\ell -1}$,  we have $r'(1/n)> 0 $ for sufficiently large $n$. 
\end{example}

\begin{remark}
Assumption~\ref{asm:schemes} implies that  $\tm{\Rew}>1$ whenever $\tm{R}>0$, thus implying that the largest possible price growth  must satisfy $\tm{\price} < \TM{\price}{t-1}/\delta = (1+r) \cdot \TM{\price}{t-1}$  where $\delta=1/(1+r)$ and $r$ is the risk-free rate. Intuitively, the return rate for just holding the token cannot beat the risk-free rate, as one might expect when looking for equilibria that keep all users and validators engaged (viability). Whether blockchains can achieve return rates competitive with the inflation is studied in \cite{felez2021better} in relation to policies adopted by some of the current blockchains, though without providing analytical results. 
\end{remark}

\subsection{Generic symmetric equilibria}
In the following we focus on the interesting case of prices having a constant multiplicative growth, which includes stable prices as a special case. 

\begin{definition}[$\gamma$-stable prices]\label{def:prices-gamma-stable} We say that the prices (of some equilibrium under consideration) are $\gamma$-stable if, 
for all $t\geq 1$, they satisfy
\begin{align}\label{eq:price-growth}
	\tm{\price} = \frac{1}{\gamma} \cdot \TM{\price}{t-1} && \gamma \neq \delta , \ \gamma >0\ . 
\end{align}
\end{definition}
Note that stable prices satisfy the above condition with $\gamma=1$, while $\gamma>1$ and $\gamma<1$ correspond to decreasing and increasing prices, respectively. 
\newcommand{\sertofee}{\mathrm{Ser2Fees}}
\newcommand{\rewtostake}{\mathrm{Rew2Stake}}

We next define a class of generic symmetric equilibria which, as we prove below, captures all symmetric equilibria with $\gamma$-stable prices.

\begin{definition}[generic symmetric equilibrium]\label{def:symm-eq-general}
	A symmetric equilibrium is generic if the following conditions hold for all $t\geq 1$ and constants $\kappa_S, \kappa_R,\rewtostake, \sertofee\in\mathbb{R}$: 
	\begin{enumerate}
		\item \label{def:symm-general-serfees} The monetary amount of tokens that users hold (and use) is proportional to the value of the service  offered by the system. That is, the  service to fees ratio is constant,
		\begin{align*}
			\tm{\sertofee} := \frac{\tm{S}}{m\tm{\token_U}\cdot \tm{\price}} = \sertofee\ .  
		\end{align*}
	\item \label{def:symm-general-rewstake}  The  amount of tokens staked by validators  is proportional to the rewards offered by the system. That is, the rewards to stake ratio is constant, 
	\begin{align*}
		\tm{\rewtostake} := \frac{\tm{R}}{n\tm{\token_V}} = \rewtostake\ .  
	\end{align*}
	\item The prices satisfy 
	\begin{align*}
		\TM{\price}{t-1} 	=&  \delta \cdot \tm{\price}\cdot (1 + \rewtostake\cdot \kappa_R) 
		\\ =&  \delta \cdot \tm{\price}\cdot \sertofee\cdot \kappa_S 
	\end{align*}
where constants $\kappa_R$ and $\kappa_S$ depend only on the number of users $m$,  the number of validators $n$, the rewards sharing scheme, and on the service fee scheme. 
		\item Each user $i$ starts round $t$ with the same token holding $\tm{\token_U}$, uses all its token current holding  for the service  ($\tm{u_i} = \tm{\token_U}$),  and buys new tokens needed for the next round accordingly   ($\tm{b_i}=\TM{\token_U}{t+1}$).
	\end{enumerate}
\end{definition}

Note that generic symmetric equilibria provide additional structure to the definition of symmetric equilibria.
The following theorem says that, when considering $\gamma$-stable prices (Definition~\ref{def:prices-gamma-stable}), we can restrict ourselves to generic symmetric equilibria without loss of generality. In particular, this holds true   for stable prices.

\begin{theorem}\label{tH:prices_implies_generic}
	Any symmetric equilibrium for the reward and service fee schemes in \eqref{eq:general-schemes}  with $\gamma$-stable prices  is a generic symmetric equilibrium with constants 
 \begin{align*}
     \kappa_R = \frac{n-1}{n} \cdot r'(1/n) && \text{and} &&  \kappa_S = g(n) \cdot \frac{m-1}{m} \cdot s'(1/m) \ . 
 \end{align*}
 This in particular holds true for the case of stable prices. 
\end{theorem}

 Note that stable prices require the following \emph{specific} constants in the two ratios involving the tokens:\begin{align}\label{eq:constants-stable-case}
    \sertofee = 1/(\delta \kappa_S) && \rewtostake = (1 - \delta)/ (\delta\kappa_R)
\end{align}
thus implying that it must hold $\kappa_R>0$ and $\kappa_S>0$. Theorem~\ref{tH:prices_implies_generic} then implies the following. 

\begin{corollary}\label{cor:stable-price-equilibria}
Any generic symmetric equilibrium (and thus any symmetric equilibrium for the reward and service fee schemes in \eqref{eq:general-schemes}) with stable prices must have the following token holdings:
\begin{align*}
    \tm{\token_U} = \frac{\tm{S}}{m} \cdot \delta \cdot \kappa_S > 0 && \text{and} && \tm{\token_V} = \frac{\tm{R}}{n} \cdot \frac{\delta}{1 - \delta} \cdot \kappa_R > 0 
\end{align*}
for any nonnegative $\tm{S}$ and $\tm{R}$.
\end{corollary}

\begin{remark}[higher rewards improve PoS security]\label{rem:rewards-PoS-security}  Security of PoS protocols relies on the total amount of  staked tokens -- the higher, the better, as an attacker needs at least a fraction of this amount to succeed.  The results above and Item~\ref{def:symm-general-rewstake} in Definition~\ref{def:symm-eq-general} indicate that the amount of staked tokens at equilibrium is \emph{proportional} to the rewards. Therefore, increasing rewards improves the system's security, a strategy even suggested in some implementations, such as Polkadot \cite{burdges2020overview}.  
    
\end{remark}

\subsection{The implications of no buybacks and Quantative Rewarding}\label{sec:nobuyback}
In the model considered so far, the system is allowed to buy or to sell the additional tokens required during each subround~II to match demand with supply. In this section, we refine a concept of equilibrium, by requiring that the system  does not have to buyback tokens (and perhaps does not sell tokens either if demand and supply match at equilibrium). We stress that this refers only to the ``spot market'' (subround~II) where tokens are exchanged for \emph{money}. Indeed, these Lagos-Wright ``two-stage'' models have a simple implementation (subround~I): The system either burns some fees (when rewards are less than the paid fees) or mints additional new tokens for rewards (when rewards are more than paid fees). Note that  \emph{no monetary reserve} is needed for these operations.

\begin{definition}[no buyback]\label{def:nobuysell}
	We say that a symmetric  equilibrium satisfies the no buyback condition if no additional tokens are bought at any round by the system, that is,   $m\tm{b_U}\geq n\tm{s_V}$ for all $t$. Additionally, we say that the no buyback condition holds tightly if demand matches supply, that is, $m\tm{b_U} = n\tm{s_V}$. 
\end{definition}

Note that the quantities $m \cdot b_U$ and $n \cdot s_V$ correspond to the token \emph{demand} and \emph{supply} in the ``market'' subround~II, respectively. As we formally prove below, the demand and supply are determined by the service value and rewards, respectively:
\begin{itemize}
    \item \emph{Demand} ($m  \tm{b_U}$ from users) decreases if the next round's \emph{service value} decreases (Equation~\eqref{eq:demand} below).
    \item \emph{Supply} ($n  \tm{s_V}$ from validators) decreases with the next round's \emph{rewards} (Equation~\eqref{eq:supply} below). 
\end{itemize}
This naturally suggests the following \emph{quantitative rewarding} approach to avoid buybacks:
\begin{enumerate}
    \item (\emph{demand decreases $\Rightarrow$ increase rewards}.) Whenever the next round's service value ``decreases significantly,'' in order to adjust supply to the decreased demand, we \emph{increase} the next round's rewards. This will make rewards more attractive for the validators, who then stake more tokens (and thus sell less on the market). 
    \item (\emph{demand increases $\Rightarrow$ decrease rewards}.) Whenever the next round's service value ``increases significantly,'' we can \emph{decrease} the next round's rewards. This will have the opposite effect on validators, who will sell more tokens on the market. 
\end{enumerate}
The first case is necessary to avoid buybacks (satisfy Definition~\ref{def:nobuysell}). The second operation turns out to be useful for (i) restoring payments to a lower value when possible and (ii) satisfying the no buybacks condition \emph{tightly}, meaning that demand never exceeds supply ($m  b_U = n  s_V$). This stronger condition has the advantage that the system does not even need to implement a minting mechanism for users to buy tokens (see Remark~\ref{rem:tight-nobuyback} below).

The following theorem  establishes an equivalence between the no buybacks condition and the above quantitative rewarding (rewards increase). In particular, it quantifies the minimum rewards for which it is possible to avoid buybacks at a given round. 

\begin{theorem}\label{th:gen-symm-equilibrium-no-buy-back}
	In any generic symmetric equilibrium, the no buyback condition holds at round $t$ if and only if the rewards satisfy the following condition 
	\begin{align}\label{eq:th:gen-symm-equilibrium-no-buy-back}
		\TM{R}{t+1} \geq \tm{R}\cdot(1+a)   - b \cdot \TM{S}{t+1} \cdot (\delta \Rew)^{t+1}
	\end{align}
where 
\begin{align*}
	a = \rewtostake \cdot n \cdot r(1/n)\ ,  && b = \frac{\rewtostake}{\sertofee}\cdot \frac{1}{\TM{\price}{0}}\ ,  && \Rew = 1+\rewtostake\cdot \kappa_R \ . 
\end{align*}
\end{theorem}
\begin{proof}
	Let us first observe that the total amount of  tokens bought by the users (demand) is 
	\begin{align}\label{eq:demand}
		m\tm{b_U} = m\TM{\token_U}{t+1} = \frac{\TM{S}{t+1}}{\TM{\price}{t+1}\cdot \sertofee} = \frac{\TM{S}{t+1}\cdot (\delta \Rew)^{t+1}}{\TM{\price}{0}\cdot \sertofee}
	\end{align}
where $\Rew = 1+\rewtostake\cdot \kappa_R$. 
The total amount of tokens sold by the validators (supply) is 
\begin{align}\label{eq:supply}
	n\tm{s_V} & =   n \left(\tm{\token_V} - \TM{\token_V}{t+1} + \tm{R}\cdot r(1/n)\right)  = \frac{\tm{R}}{\rewtostake}  - \frac{\TM{R}{t+1}}{\rewtostake}  + \tm{R} \cdot n \cdot r(1/n) \ . 
\end{align}
Hence, the no buyback condition $m\tm{b_U}\geq n\tm{s_V}$ is equivalent to
\begin{align*}
	\frac{\tm{R}}{\rewtostake}  - \frac{\TM{R}{t+1}}{\rewtostake}  + \tm{R} \cdot n \cdot r(1/n) \leq  \frac{\TM{S}{t+1}\cdot (\delta \Rew)^{t+1}}{\TM{\price}{0}\cdot \sertofee}
\end{align*}
that is 
\begin{align*}
	\tm{R}  - \TM{R}{t+1} + \tm{R} \cdot n \cdot r(1/n) \cdot \rewtostake \leq  \frac{\rewtostake}{\TM{\price}{0}\cdot \sertofee} \cdot  \TM{S}{t+1}\cdot (\delta \Rew)^{t+1}\ . 
\end{align*}
By rearranging the terms, the theorem follows. 
\end{proof}

The above theorem provides a recursive (algorithmic) formula for the rewards, which leads to our mechanism described in Section~\ref{sec:implementation-single-token} below. By unfolding the recursion \eqref{eq:th:gen-symm-equilibrium-no-buy-back} in  same theorem, we get the following bound on the growth of the rewards necessary to guarantee the no buyback. 
\begin{corollary}\label{cor:reward-growth-single}
    The minimal rewards satisfying the no buyback condition (tightly) are equal to 
    \begin{align*}
        \TM{\Roptsingle}{t+1} = (1 + a)^t \cdot \left(\TM{R}{0} - b \sum_{\tau=1}^{t} \TM{S}\tau{}\cdot \left(\frac{\delta\Rew}{1+a}\right)^\tau \right)\ ,  && t \geq 1 
    \end{align*}
    for any initial reward $\TM{R}{0}$, and for constants $a$, $b$, and $\Rew$ as in Theorem~\ref{th:gen-symm-equilibrium-no-buy-back}.
\end{corollary}
We stress  that  the \emph{minimal} rewards satisfying the no buyback condition must   satisfy the no buyback condition \emph{tightly}, and thus must coincide with the   rewards $\tm{\Roptsingle}$ in Corollary~\ref{cor:reward-growth-single}. 

\begin{remark}[avoid users minting with tight no buyback]\label{rem:tight-nobuyback}
Observe that  the no buyback condition   only requires supply to not exceed demand (Definition~\ref{def:nobuysell}). Hence, it is possible that (users') demand exceeds (validators') supply  
($
    m \tm{b_U} - n \tm{s_V}> 0
$).
In this case, the system has to provide this difference of tokens to the market (users). On the one hand, it allows the system to  make  some further profit. 
On the other hand, 
there are scenarios in which this might be difficult or undesirable (typically, it requires the system to implement a ``special'' minting mechanism that allows users to buy tokens). 
In such cases, we require a \emph{tight} version of the no buyback condition (Definition~\ref{def:nobuysell}) in which demand matches supply. 
The analysis of this tight no buyback condition is given by Corollary~\ref{cor:reward-growth-single}, while the formula in Theorem~\ref{th:gen-symm-equilibrium-no-buy-back} with equality gives the recursion for the rewards.  
\end{remark}

As we discuss below, implementing quantitative rewarding can be difficult due to the possibility of an uncontrolled rewards growth.  Indeed, Corollary~\ref{cor:reward-growth-single} states that the initial rewards $\TM{R}{0}$ cannot exceed a certain value unless rewards grow exponentially. This value, critically, depends on future service values. In Section~\ref{sec:implementation-single-token}, we describe an implementation under the assumption that the initial rewards $\TM{R}{0}$ can be set properly by the system (either having full knowledge about future service values, a good estimate, or some partial control over them).

\begin{remark}[QR vs fees burning]
\label{remark:burning}
	In practice,  the system collects fees paid in round $t$ and uses part of them as rewards in the same round. Validators receive $\tm{R} \cdot n \cdot r(1/n)$ tokens in total, defining the ``rewards-over-fees'' ratio as $\tm{ROF} := \frac{\tm{R} \cdot n \cdot r(1/n)}{m \tm{\token_U}}$. Systems that burn a constant fraction of fees and set the remainder as rewards maintain a constant $\tm{ROF} < 1$. Quantitative rewarding is more general: (1) $\tm{ROF}$ can vary over time, implying a nonconstant burn rate, and (2) $\tm{ROF} > 1$ allows monetary expansion, introducing extra reward tokens (still avoiding user minting according to Remark~\ref{rem:tight-nobuyback}). 
\end{remark}

\subsubsection{Main dilemma: Uncontrolled growth of rewards}\label{sec:dilemma-rewards-growth}
The bound in Corollary~\ref{cor:reward-growth-single} suggests that an \emph{exponential} growth of the rewards over time is necessary if the system starts with a too high initial rewards or  the service value is too low. 
Hence, although larger rewards might be desirable to improve security (Remark~\ref{rem:rewards-PoS-security}), some care is needed in case of ``shocks'' in the service value. 

    In the following experiment, we consider  constant service value and rewards that are set to the minimal value necessary in order to have no buyback (Theorem~\ref{th:gen-symm-equilibrium-no-buy-back}), and also impose the rewards to not be below some minimum value (set to 0.1 for the sake of exposition). For the sake of simplicity, we consider the rewards and cost sharing schemes in \eqref{eq:simpler-schemes} and  $g(n)=1$, and observe the following:
\begin{itemize}
    \item Sufficiently high service value can avoid rewards explosion (Figure~\ref{fig:impl-stable-const-S} compares two different service values under the same conditions). The necessary increase in the service value may simply be not possible due to inherent technological limitations and other external factors. 
    \item  Sufficiently small initial rewards also avoid the rewards explosion (Figure~\ref{fig:impl-stable-const-S-rewards} compares two initial rewards under the same conditions). Rewards however cannot be arbitrarily small since they must cover the costs of validators and be competitive against other source of investments. 
    \item A higher risk free rate may also trigger an exponential growth (Figure~\ref{fig:impl-stable-const-S:smaller-delta} shows that for smaller $\delta$ we need a smaller initial reward, or a larger service value, to avoid this explosion). 
    \item Allowing the token price to decrease does avoid the rewards explosion that instead occur with stable prices, once we consider the rewards in money (Figure~\ref{fig:impl-non-stable-const-S}). Decreasing prices are however not desirable, and perhaps one may want increasing prices, which turn out to make the monetary reward explosion even more severe. 
\end{itemize}

    The above uncontrolled growth represents a potential ``death spiral'' stemming from necessary equilibrium conditions. First, it leads to impractical implementation due to numerical explosion in ledger validators accounts (thus violating feasibility). Second, it  may even render such equilibria nonexistent for some parameter combinations.  Note that adapting the rewards requires knowledge about shocks far in the future. Indeed, Theorem~\ref{th:gen-symm-equilibrium-no-buy-back} states that rewards cannot be reduced arbitrarily from one step to the next. Such powerful oracles are generally not available, and solutions based on them also violate feasibility.

    \begin{figure}
        \centering
        \includegraphics[scale=.5]{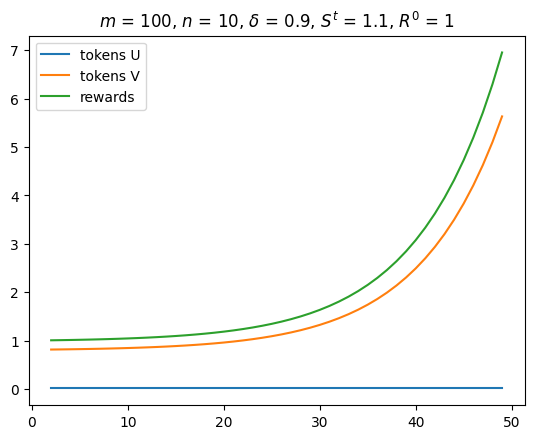}
        \includegraphics[scale=.5]{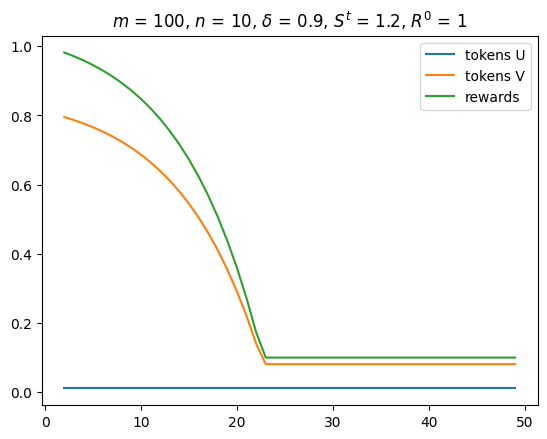}
        \caption{The effect of service values on rewards growth when no buyback condition: a small service value (left) can trigger an exponential increase in both rewards and validators staked amounts, while a slightly bigger service value (right) can remove this growth.}
        \label{fig:impl-stable-const-S}
    \end{figure}
    
    \begin{figure}
        \centering
        \includegraphics[scale=.5]{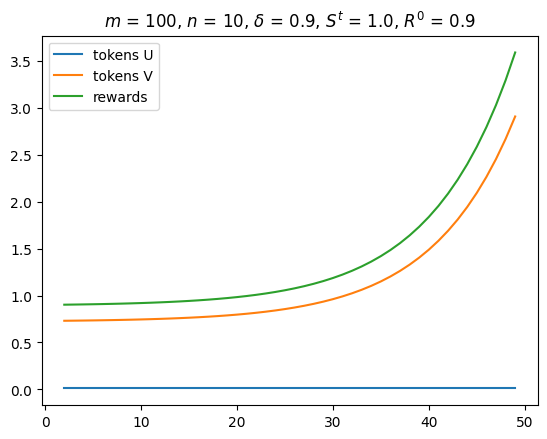}
        \includegraphics[scale=.5]{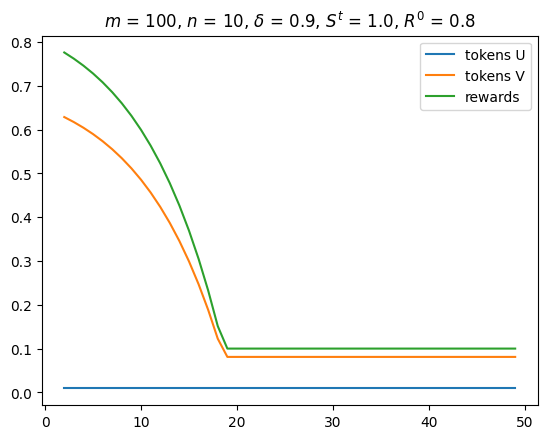}
        \caption{A too high initial rewards (left) causes an exponential growth, as opposed to a smaller initial rewards (right) leading to stable rewards and staked tokens (note the different scale of the y-axis in the two plots). In both cases, we fix a minimum value of $0.1$ for the rewards, which is the flat part of the green line in the right plot.}
        \label{fig:impl-stable-const-S-rewards}
    \end{figure}
    
    \begin{figure}
        \centering
        \includegraphics[scale=.5]{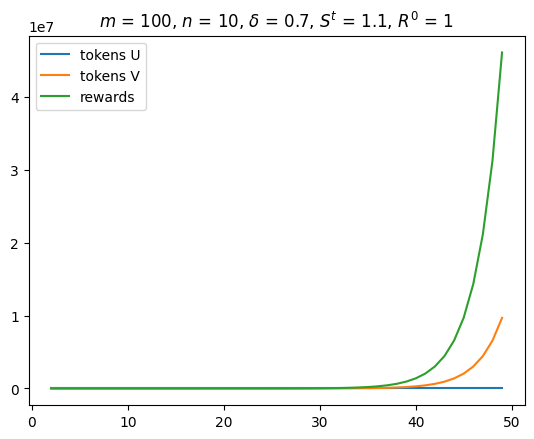}
        \includegraphics[scale=.5]{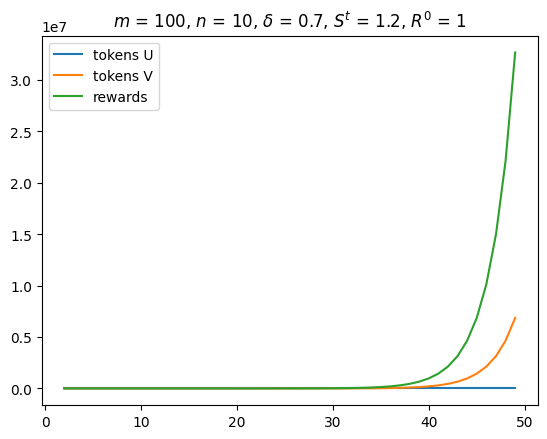}
        \caption{The effect of risk free rate ($\delta$) on rewards growth when no buyback condition: a smaller $\delta$  might trigger an exponential increase in both rewards and validators staked amounts (compare the right picture here with the right picture  in Figure~\ref{fig:impl-stable-const-S}).}
        \label{fig:impl-stable-const-S:smaller-delta}
    \end{figure}

\begin{figure}
        \centering
        \includegraphics[scale=.5]{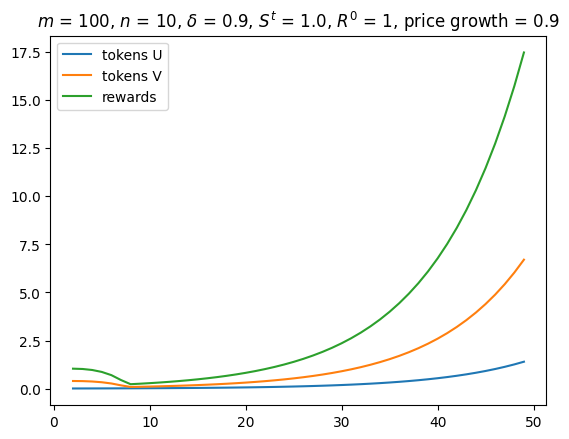}
        \includegraphics[scale=.5]{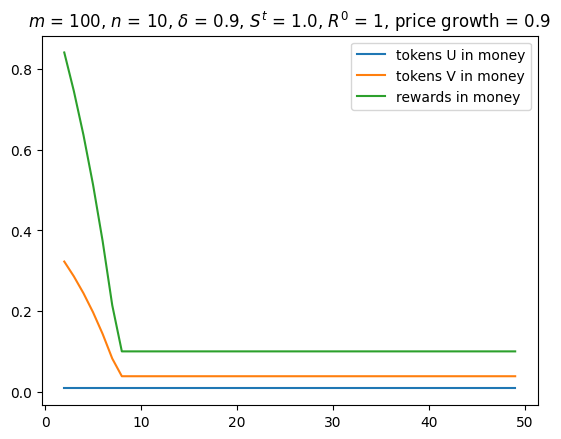}
        \caption{The effect of non-stable prices on rewards and validators staked tokens. We set decreasing prices according to the price growth parameter $\delta\Rew=0.9$. Though the rewards and staked tokens increase exponentially (left) their value  expressed in money -- rewards and staked tokens market cap -- does not (right).}
        \label{fig:impl-non-stable-const-S}
    \end{figure}

\subsection{Policy that accommodates and implements the stable price equilibrium}\label{sec:implementation-single-token}
In this section, we consider the fundamental question of how one can compute and implement an equilibrium with stable prices satisfying the no buyback condition. Stable prices require to keep a precise  amount of tokens of users (on one side) and of tokens of validators (on the other side) given by Corollary~\ref{cor:stable-price-equilibria}. The no buyback condition requires the system to set the rewards so to satisfy the condition in Theorem~\ref{th:gen-symm-equilibrium-no-buy-back}.
For the sake of exposition, we next describe an implementation for the reward and service fee schemes in \eqref{eq:general-schemes}. We should point out that this implementation does not specify how to set the initial rewards $\TM{R}{0}$ which is the crucial point to avoid rewards explosion (Corollary~\ref{cor:reward-growth-single}).   

\paragraph{Given parameters and constraints.}
\begin{enumerate}
	\item The risk free rate $r$ and thus the discount factor $\delta = 1/(1+r)$.
	\item The total number $m$ of users and the total number $n$ of validators (both constant over time). 
	\item There is no need for the system to buyback tokens  (Definition~\ref{def:nobuysell}).  
\end{enumerate}
\paragraph{System parameters.} The system design depends on essentially two parameters: 
\begin{enumerate}
	\item Total value of the service  $\tm{S}$ which is divided among the users at each step $t$.
	\item Total amount of rewards  $\tm{R}$ which is distributed to the validators at each step $t$.  
\end{enumerate}

\paragraph{Assumption.} By the beginning of Subround~II (buy or sell) of a current round $t$, each user knows the service value $\TM{S}{t+1}$ at the next round, and each validator knows the total rewards $\TM{R}{t+1}$ available in the next round.

\paragraph{Suggested equilibrium description and properties.} At each round $t\geq 0$, 
\begin{enumerate}
	\item Users strategies: 
	\begin{enumerate}
		\item \label{itm-stable-implement-tkU}
        Start with $\tm{\token_U}=\frac{\tm{S}}{m} \cdot \delta \cdot \kappa_S$ tokens, where $\kappa_S= g(n) \cdot \frac{m-1}{m}\cdot s'(1/m)$.
		\item  Spend all these tokens to get a service value equal to $\tm{S} \cdot s(1/m)$. 
		\item Given the next round service value $\TM{S}{t+1}$, buy   $\TM{b_{U}}{t} = \TM{\token_{U}}{t+1} = \frac{\TM{S}{t+1}}{m} \cdot \delta \cdot \kappa_S$  tokens required for the next round according to Item~\ref{itm-stable-implement-tkU}.
	\end{enumerate}  
	\item Validators strategies: 
	\begin{enumerate}
		\item \label{itm-stable-implement-tkV} Start with $\tm{\token_V} =  \frac{\tm{R}}{n} \cdot \frac{\delta}{1 - \delta} \cdot \kappa_R$ tokens, where $\kappa_R = \frac{n-1}{n} \cdot r'(1/n)$.
  
		\item  Stake all these tokens to get  $\tm{R} \cdot r(1/n)$ new tokens as reward. 
		\item Given the rewards $\TM{R}{t+1}$ of next round, sell this amount of tokens (if negative, buy tokens): 
		$
			\tm{s_V} =  \left(\tm{R} - \TM{R}{t+1}\right)\cdot  \frac{\delta}{1 - \delta}\cdot \kappa_R + \tm{R}\cdot r(1/n)   
		$,
	which yields  the required $\TM{\token_V}{t+1}=\frac{\TM{R}{t+1}}{n} \cdot \frac{\delta}{1 - \delta} \cdot \kappa_R$ tokens for next round $t+1$ according to Item~\ref{itm-stable-implement-tkV}. 
	\end{enumerate} 
	\item  System actions: 
	\begin{enumerate}
		\item Given the next round service value $\TM{S}{t+1}$, set the next total rewards $\TM{R}{t+1}$ so that 
		\begin{align*}
			m \cdot \TM{S}{t+1} \cdot \delta \cdot \kappa_S \geq n \cdot  \left(\tm{R} - \TM{R}{t+1}\right)\cdot  \frac{\delta}{1 - \delta}\cdot \kappa_R + n \cdot \tm{R}\cdot r(1/n)  
		\end{align*}
	which guarantees the no buyback condition, that is, $m \cdot \tm{b_U}\geq n \cdot \tm{s_V}$. 
	\item Announce the next  rewards $\TM{R}{t+1}$ before the token buy or sell (Subround~II) begins.
	\item Sell the required amount of tokens, $m \TM{\token_U}{t+1} - n \tm{s_V}$, to match demand during Subround~II. 
	\end{enumerate}
	\item Prices: all tokens are exchanged at a stable price $\tm{\price}=1$.
\end{enumerate}

\section{A two-token model}\label{sec:two-token-model}
\newcommand{\tokenpay}{\mathbf{A}}
\newcommand{\tokenstake}{\mathbf{B}}
\newcommand{\tokenGEN}[1]{\mathbf{#1}}

\newcommand{\tkA}[1]{{#1}_{\tokenpay}}
\newcommand{\tkB}[1]{{#1}_{\tokenstake}}

\newcommand{\tokenA}{\tokenpay}
\newcommand{\tokenB}{\tokenstake}

\newcommand{\priceA}{\tkA{\price}}
\newcommand{\priceB}{\tkB{\price}}

\newcommand{\RewA}{\tkA{\Rew}}
\newcommand{\RewB}{\tkB{\Rew}}

\newcommand{\RRA}{\tkA{R}}
\newcommand{\RRB}{\tkB{R}}

\newcommand{\sellA}[1]{\tkA{s_{#1}}}
\newcommand{\sellB}[1]{\tkB{s_{#1}}}

\newcommand{\buyA}[1]{\tkA{b_{#1}}}
\newcommand{\buyB}[1]{\tkB{b_{#1}}}

\newcommand{\useA}[1]{\tkA{u_{#1}}}
\newcommand{\useB}[1]{\tkB{u_{#1}}}

\newcommand{\inflation}{{\mathcal{I}}}
\newcommand{\inflationA}{\tkA{\mathcal{I}}}
\newcommand{\inflationB}{\tkB{\mathcal{I}}}

In this section, we consider a natural variant of the previous (single-token) model to the case in which we have two tokens, token $\tokenA$ and $\tokenB$,  whose use and purpose is as follows:
\begin{itemize}
    \item token $\tokenB$ is used by the users to pay and get the service, and 
    \item token $\tokenA$ is used by the validators for staking to get rewarded with more tokens (of both types).
\end{itemize}
A number of existing systems follow a similar two-token scheme (see Figure~\ref{fig:two-token-systems} for some examples) and a model based on the similar assumptions has been recently given by \cite{Dimitri2023}. 
Both tokens can be exchanged in the spot market as before, and thus we have a price for each token (see  Figure~\ref{fig:notation-two-token} for some additional notation for the two-token model). The corresponding utilities are naturally given by  
    \begin{align}\label{eq:usr-utility-istant-two}
     \disc{U_i} =  \sum_{t=0}^\infty \delta^t \cdot \tm{U_i}\ ,  && 
        \tm{U_i} = & \tm{S_i} \cdot g(n) - \tm{\buyA{i}} \cdot \tm{\priceA} - \tm{\buyB{i}} \cdot \tm{\priceB} \ . 
        \\
        \label{eq:val-utility-istant-two}
        \disc{V_j} =  \sum_{t=0}^\infty \delta^t \cdot \tm{V_j}\ ,  && \tm{V_j} = & \tm{\sellA{j}} \cdot \tm{\priceA} + \tm{\sellB{j}} \cdot \tm{\priceB}- v \ . 
    \end{align}

\begin{figure}\centering
	\begin{tabular}{|c|c|c|c|c|}
		\hline
		System & token $\tokenA$ & purpose & token $\tokenB$ & purpose \\\hline
		\hline
		DFINITY \citep{dfinity2022internet} & ICP & staking \& rewards & Cycle & computation (service) \\
		\hline
		NEO \citep{NEO-dual-token} & NEO & staking & GAS & pay transactions  \\
		\hline
		Axie Infinity \citep{Axie-dual-token} & AXS & governance and staking & SLP & play (breeding Axie)  
        \\
		\hline
		IOTA \citep{saa2023iota} & IOTA &  staking \& rewards & Mana & access ledger  \\
		\hline
	\end{tabular}
	\caption{Examples of two-token systems.}
	\label{fig:two-token-systems}
\end{figure}

\begin{figure}
    \centering
\block{
\begin{itemize}
    \item $\tm{R} = (\tm{\RRA}, \tm{\RRB})$ is the reward at time $t$.
    \item $\tm{s}_j = (\tm{\sellA{j}}, \tm{\sellB{j}})$ is the selling strategy of validator $j$.
    \item  $\tm{b}_i = (\tm{\buyA{i}}, \tm{\buyB{i}})$ is the buying strategy of user $i$. 
    \item $\tm{u_i}$ is the amount of token $\tokenA$ used by user $i$ to pay for the service.
    \item $\tm{\token_p}=(\tm{\tokenA_p}, \tm{\tokenB_p})$ are the token holdings of a generic player $p$ (a user or a validator).
    \item $\tm{\price}=(\tm{\priceA}, \tm{\priceB})$ are the prices of the two tokens. 
\end{itemize}
}   
\caption{Notation for the two-token model.}
    \label{fig:notation-two-token}
\end{figure}

\newcommand{\RRAj}{R_{\tokenA j}}
\newcommand{\RRBj}{R_{\tokenB j}}

For the sake of exposition, we consider the simpler reward sharing and payment schemes \eqref{eq:simpler-schemes}, where token $\tokenA$ is used for staking, and each validator $j$ is rewarded with quantities $\tm{\RRAj}$ and $\tm{\RRBj}$ of tokens $\tokenA$ and $\tokenB$, respectively. Users get a service value depending on the amount of tokens $\tokenB$ they use, i.e., they pay to the system. Thus, we have  
\begin{align}\label{eq:two-token-rewards-simple}
\tm{S_i} = \tm{S} \cdot \left(\frac{\tm{u_i}}{\sum_{k} \tm{u_k}}\right)\ ,  && 
	\tm{\RRAj} = \tm{\RRA} \cdot \left(\frac{\tm{\tokenA_j}}{\sum_v \tm{\tokenA_v}} \right)\ ,  &&
	\tm{\RRBj} = \tm{\RRB} \cdot \left(\frac{\tm{\tokenA_j}}{\sum_v \tm{\tokenA_v}} \right) \ . 
\end{align}
The validators' selling constraint for the single token \eqref{eq:val-sell-constraint} extends naturally into 
\begin{align}
	\label{eq:val-sell-constraint-A}
	\TM{\tokenA_j}{t+1}  =  \tm{\tokenA_j} + \tm{\RRAj}  - \tm{\sellA{j}} \geq 0 \\
	\label{eq:val-sell-constraint-B}
	\TM{\tokenB_j}{t+1} =  \tm{\tokenB_j} + \tm{\RRBj}  - \tm{\sellB{j}}\geq 0
\end{align}
Similarly, the users' selling constraint for the single token \eqref{eq:usr-sell-constraint} extends as follows:
\begin{align}
    \TM{\tokenA_i}{t+1} =  \tm{\tokenA_i} + \buyA{i} \geq 0 \label{eq:usr-sell-constraint-A}
    \\
    \TM{\tokenB_i}{t+1} =  \tm{\tokenB_i} - \tm{u_i} + \buyB{i} \geq 0
    \label{eq:usr-sell-constraint-B}
\end{align}
where the asymmetry is due to the fact that only token $\tokenB$ is used for getting the service, by paying  $\tm{u_i}$ tokens.
We next generalize the condition of  (symmetric)  equilibrium (Definition~\ref{def:weakly-symmetric-equilibrium} and Remark~\ref{rem:positive-token-holding}) by requiring a ``minimal'' set of selling constraints to be  non-strict (see Remark~\ref{re:non-strict-two-tokens} below).

\begin{definition}[symmetric equilibrium two tokens]\label{def:symm-two-tokens}
Consider a system policy given by \eqref{eq:two-token-rewards-simple}, and a triple of users' strategies, validators' strategies, and  prices, $\{(\tm{u_i}, \tm{b_i}), \tm{s_j},  \tm{\price}\}$, 
    such that the corresponding selling constraints of users \eqref{eq:usr-sell-constraint-A}-\eqref{eq:usr-sell-constraint-B} and of validators \eqref{eq:val-sell-constraint-A}-\eqref{eq:val-sell-constraint-B} are satisfied. We say that such a triple is an \emph{equilibrium} if, for every user $i$ and every validator $j$
    \begin{enumerate}
       \item  $(\tm{u_i}, \tm{b_i})$ maximizes the discounted utility \eqref{eq:usr-utility-istant-two} of user $i$,  given all other strategies  \item  $\tm{s_j}$ maximizes the discounted utility \eqref{eq:val-utility-istant-two} of validator $j$,  given all other strategies
       
        \item the selling constraint of token $\tokenB$ are non-strict for user $i$ (resp.,  token $\tokenA$   for validator $j$), and thus the corresponding token holdings are strictly positive, that is, $\tm{\tokenB_i}>0$ and $\tm{\tokenA_j}>0$
    \end{enumerate} 
    Moreover, we say that such an equilibrium is \emph{symmetric} if the token holdings of all users are the same, and similarly, if all token holdings of all validators are the same. In particular,  we have
    \begin{align}\label{eq:symmetric--token-holdins-two}
        \tm{\tokenB_i}= \tm{\tokenB_U}>0 && \tm{\tokenA_j} = \tm{\tokenA_V}>0 
    \end{align}
    for all users $i$ and for all validators $j$ and for all $t$.
\end{definition}

\begin{remark}[minimal strict selling constraints]\label{re:non-strict-two-tokens}
    Note that in the definition above,  each type of player -- validator or user  -- has non-strict selling constraint only in its own ``main purpose'' token (see Equation~\ref{eq:symmetric--token-holdins-two}). This  corresponds to the natural requirement of continuous participation in staking and in accessing and paying for the service.  Furthermore, the price analysis below implies that equilibria with certain desired prices are impossible if some of the other selling constraints is also strict. 
\end{remark}

The proof of the next two lemmas is given in  Section~\ref{sec:two-token-analysis}. 
 Next lemma is the generalization of Lemma~\ref{le:prices-restricted} to two tokens, and it provides conditions in the prices based on the validators' strategies.

\begin{lemma}[validators part]\label{lem:prices-two-tokens}
In the two-token model, the corresponding prices at any  symmetric equilibrium must satisfy the following conditions. 
    If the selling constraints of $\tokenB$ are non-strict, then 
    $
    \TM{\priceB}{t-1}  = \delta  \cdot \tm{\priceB} \ . 
$
If the selling constraints of $\tokenA$ are non-strict, then
\begin{align}\label{eq:val-priceA-two-token}
    \TM{\priceA}{t-1} & = \delta  \cdot \tm{\priceA} \cdot (1 +\tm{\inflationA}) +  \delta \cdot \tm{\priceB} \cdot \tm{\inflationB} \ ,
\end{align}
where
\begin{align}\label{eq:two-token-inflations}
	\tm{\inflationA}:= \frac{\tm{\RRA}}{n\tm{\tokenA}_{V}} \frac{n-1}{n} \ ,&& \tm{\inflationB} := \frac{\tm{\RRB}}{n\tm{\tokenA}_{V}} \frac{n-1}{n} \ .
\end{align}
and $\tm{\tokenA_V}$ is the token holding (staking) of any validator at this equilibrium.    
\end{lemma}

Note that in absence of token $\tokenB$, we recover the result in Lemma~\ref{le:prices-restricted} for a single token. In this case, only token $\tokenA$ is used for staking and for rewards, and thus $\tm{\RRB}=0=\tm{\tkB{\mathcal{I}}}$. The two quantities in $\eqref{eq:two-token-inflations}$  have an intuitive meaning as they correspond to the total amount of tokens rewarded -- in each of the two tokens -- over the total amount of staked tokens -- the latter comprise only tokens  $\tokenA$. Note  that these two expressions in $\eqref{eq:two-token-inflations}$ are \emph{not} symmetric in the two tokens, as the staking token $\tokenA$ appears in both denominators. 

We next consider how the users' buying strategies relate to the prices.

\begin{lemma}[users part]\label{lem:prices-two-tokens-usr-side}
    In the two-token model,  the corresponding prices at any  symmetric equilibrium must satisfy the following conditions. If the buying constraints of $\tokenB$ are non-strict, then 
    \begin{align*}
        \TM{\priceB}{t-1} = \delta \cdot \begin{cases}
        \tm{\Ser} & \text{if  } \tm{u_i} = \tm{\tokenB_U} \\
        \tm{\priceB} & \text{if  } \tm{u_i} < \tm{\tokenB_U} 
    \end{cases}
    && \tm{\Ser} = \frac{\tm{S}}{m\tm{\tokenB_U}} \cdot g(n) \cdot \frac{m-1}{m} \ . 
    \end{align*}
    Moreover, if the selling constraints of $\tokenA$ are non-strict, then 
$
        \TM{\priceA}{t-1} = \delta  \cdot \tm{\priceA} . 
 $
\end{lemma}

\subsection{Generic symmetric equilibria}\label{sec:generic-symmetric-two}
\newcommand{\mixtwo}{\mathbb{MIX}_{\tokenA\tokenB}}
\newcommand{\RewAB}{\Rew_{\tokenpay\tokenstake}}
\newcommand{\sertofeeB}{Ser2Fees_{\tokenB}}
\newcommand{\rewtostakeA}{Rew_{\tokenA}2Stake}
\newcommand{\rewtostakeB}{Rew_{\tokenB}2Stake}

The following class of generic equilibria captures equilibria in the two-token model where prices of token $\tokenB$ are \emph{$\gamma$-stable}, that is, they satisfy \eqref{eq:price-growth}, thus in particular the case in which we aim at stable prices for  token $\tokenB$ used by the user to get the service. 

\begin{definition}[generic symmetric equilibrium for two tokens]\label{def:two-token-symm-eq-general}
	A symmetric equilibrium for the two token model is generic if the following conditions hold for all $t\geq 1$:
	\begin{enumerate}
		\item The service to fees ratio is constant,
		\begin{align*}
			\tm{\sertofeeB} := \frac{\tm{S}}{m\tm{\tokenB_U}\cdot \tm{\priceB}} = \sertofeeB\ .  
		\end{align*}
  \item The prices of $\tokenB$ satisfy 
		\begin{align*}
			\TM{\priceB}{t-1} =&  \delta \cdot \tm{\priceB}\cdot \sertofeeB\cdot \kappa_S 
		\end{align*}
		where constant $\kappa_S$ depends only on the number of users $m$,  the number of validators $n$, and on the service fee scheme. 
		\item Each user $i$ starts round $t$ with the same token holding $\tm{\tokenB_U}$, uses all its current holding tokens for the service  ($\tm{u_i} = \tm{\tokenB_U}$),  and buys new tokens needed for the next round accordingly   ($\tm{b_i}=\TM{\tokenB_U}{t+1}$).
	\item For the following rewards to stake ratios, 
		\begin{align*}
			\tm{\rewtostakeB} := \frac{\tm{\RRB}}{n\tm{\tokenA_V}}  
		\end{align*}
	the prices of $\tokenA$ satisfy
	\begin{align}\label{eq:price-two-tokens-generic}
		\TM{\priceA}{t-1} & = \delta  \cdot \tm{\priceA} \cdot (1 +\tm{\rewtostakeA}\cdot \kappa_R) +  \delta \cdot \tm{\priceB} \cdot \tm{\rewtostakeB}\cdot \kappa_R \ ,
	\end{align}
	where  constant $\kappa_R$ depends only on the number of users $n$,  and on the rewards sharing scheme.
	\end{enumerate}
\end{definition}

\begin{theorem}\label{th:symm-eq-two-token}
Any symmetric equilibrium for the two-token model with the reward and service fee schemes in \eqref{eq:two-token-rewards-simple} and whose prices of token $\tokenB$ satisfy \eqref{eq:price-growth} is a generic symmetric equilibrium.
	This in particular holds true for the case of stable prices for token $\tokenB$. 
\end{theorem}

\subsection{No buyback  in the two-token model}
In this section, we study the implications of no buyback in the two-token model. As for the single token setting, we aim at equilibria where the system does not need to buyback tokens of either type.

\begin{definition}[no buyback two tokens]\label{def:no-buy-back-two-tokens}
      We say that a symmetric  equilibrium in the two-token model satisfies the no buyback condition if no additional tokens (of either type) are bought at any round by the system, that is, 
	$m \tm{\buyB{i}} \geq n \tm{\sellB{j}}$ and  $m \tm{\buyA{i}} \geq n \tm{\sellA{j}}$. 
\end{definition}

The following theorem provides a bound on the rewards that depends only on the next round's service value. The main advantage of the two-token model is that rewards can be adjusted ``immediately'' to the service value fluctuations, and they need to be low only when the latter is low. The mechanism in Section~\ref{sec:stable-price-mechanism-two} implements this idea.

\begin{theorem}\label{th:generic-nobuy-back-tow}
In any generic symmetric equilibrium, there is a maximum monetary reward for validators  in terms of tokens  $\tokenB$ that can be awarded without violating the no buyback condition. In particular, it must hold 
		\begin{align*}
			\tm{\RRB}\cdot \tm{\priceB}  \leq \delta \cdot\TM{S}{t+1} \cdot  \frac{\kappa_S}{n \cdot r(1/n)}    \ . 
		\end{align*}
  This is because validators always sell all the newly rewarded $\tokenB$ tokens and users buy (only) tokens of type $\tokenB$, thus implying that validators keep all their tokens $\tokenA$ and possibly buy new ones, $\tm{\tokenA_V}\geq \TM{\tokenA_V}{t-1}$. 
\end{theorem}

\begin{proof}
	  We consider the two inequalities of the no buyback condition (Definition~\ref{def:no-buy-back-two-tokens}) separately.
	\begin{enumerate}
		\item For tokens of type $\tokenB$, we observe that, for $\TM{\priceB}{t-1} \neq \delta \cdot \tm{\priceB}$,  the users buying strategies satisfy
		$
					m\tm{\buyB{i}} = 	m\TM{\tokenB_U}{t+1}  = \frac{\TM{S}{t+1}}{\TM{\sertofeeB}{t+1}\cdot \TM{\priceB}{t+1}}  =  \TM{S}{t+1} \cdot \frac{\delta}{\tm{\priceB}} \cdot  \kappa_S\ .   
		$
		Moreover, the validators' selling strategies satisfy
		$
			\tm{\sellB{j}}  = \tm{\RRB}\cdot r(1/n) \ , 
		$
		since the validators selling constraint for tokens of type $\tokenB$ must be strict, and thus $\tm{\tokenB_i}=0$ for all $t$. 
		Therefore the no buyback condition for tokens of type $\tokenB$ is equivalent to
		$
			n \cdot \tm{\RRB}\cdot r(1/n) \leq\TM{S}{t+1} \cdot \frac{\delta}{\tm{\priceB}} \cdot  \kappa_S . 
		$
		
		\item We also require no buyback for the tokens of type $\tokenA$, which means that $\tm{\sellA{j}}=0$ because, in any generic symmetric equilibrium users never hold (and thus never buy) tokens of type $\tokenA$: If users hold tokens $\tokenA$, the second part Lemma~\ref{lem:prices-two-tokens-usr-side} implies $\TM{\priceA}{t-1}  = \delta  \cdot \tm{\priceA}$, thus contradicting  \eqref{eq:price-two-tokens-generic} in Definition~\ref{def:two-token-symm-eq-general}. Hence,
		$
			\TM{\tokenA_V}{t+1} = \tm{\tokenA_V} +  \tm{\RRA}\cdot r(1/n) + \tm{\buyA{j}}
		$ with $\tm{\buyA{j}}\geq 0$.
	\end{enumerate}
 This completes the proof. 
\end{proof}

\subsection{Stable price mechanism in two-token model}\label{sec:stable-price-mechanism-two}
We next describe a simple mechanism which implements stable prices for token $\tokenB$, while  the price of $\tokenA$ is strictly decreasing:
\begin{enumerate}
    \item In order to have stable prices for $\tokenB$ and no buyback, we set $\tm{\RRB}$ such that 
    $
        \tm{\RRB}\cdot \tm{\priceB} = \TM{S}{t+1}\cdot L
    $,
    where $L$ is the constant given by Theorem~\ref{th:generic-nobuy-back-tow} that makes the no buyback condition hold tightly. Furthermore, we set $\tm{\RRA}=0$.
    \item The (equilibrium) buying strategies  for token $\tokenA$ are ``no buy and no sell'', that is,  $\tm{\tokenA_V} = \TM{\tokenA_V}{0}$, for a suitable $\TM{\tokenA_V}{0}$ specified below.  
    From the equation of the prices \eqref{eq:val-priceA-two-token}, we obtain  (details in Appendix~\ref{app:proof-priceA-proposed-equilibrium})
    \begin{align}\label{eq:priceA-proposed-equilibrium}
	\tm{\priceA} & =  \frac{\TM{\priceA}{t-1}}{\delta} -  \TM{S}{t+1}\cdot  \frac{L }{\TM{\tokenA_V}{0}} \frac{n-1}{n^2}   \ , 
\end{align}
which implies 
$
	\tm{\priceA}  <  \frac{\TM{\priceA}{0}}{\delta^t}
    .$
    \item From the previous equation, the discounted utility of validators, if deviating by selling all tokens $\tokenA$ at some round $\tau$, is at most 
     $
     	\delta^\tau \cdot \TM{\tokenA_V}{0} \cdot \TM{\priceA}{\tau} < \TM{\tokenA_V}{0}\cdot \TM{\priceA}{0},$
 while the discounted utility for the suggested strategies  (equilibrium) equals
 \begin{align*}
\sum_{t=0}^{\infty} \delta^t \cdot (\tm{\RRB}\cdot \tm{\priceB} - v) =  \sum_{t=0}^{\infty} \delta^t \cdot(\TM{S}{t+1}\cdot L) - \frac{v}{1-\delta}\ . 
 \end{align*}
\item In order to have feasible (nonnegative) prices, we need to set the initial payments and token holdings sufficiently high. 
\end{enumerate}

\begin{theorem}
\label{th:two-token-equilibria-exist}
    For any  service value satisfying Assumption~\ref{asm:service-level} and any reward and service sharing schemes satisfying Assumption~\ref{asm:schemes}, the strategies above are a symmetric equilibrium with stable prices for token $\tokenB$ used by users for getting the service. 
    \end{theorem}
\begin{proof}
    We follow a similar argument as in \cite[Remark~2 (existence)]{decentralization_friction} based on the following three conditions. First, the discounted utilities of all players (users and validators) assume finite values (since their instantaneous utilities are proportional to $\tm{S}$). Second, these strategies at equilibrium are in some bounded interval $[\underline{\theta},\overline{\theta}]$. For users, this is evident since they buy tokens proportionally to $\TM{S}{t+1}$ and from Assumption~\ref{asm:service-level}. For the validators, they sell all $\tm{\RRB} = O(\tm{S})$ tokens $\tokenB$, they never sell or buy tokens $\tokenA$. Third, the  utilities of the players are concave due to Assumption~\ref{asm:schemes}. Finally, Theorem 4.5 in \cite{stokey1989recursive} yields the desired result (concavity implies that the necessary equilibrium conditions maximize a round-$t$ decomposition of the utility, like  \eqref{eq:val-utility-two-tokens} for validators, while Theorem 4.5 in \cite{stokey1989recursive} ensures the maximum for the decomposition indeed maximizes the corresponding player utilities, thus an equilibrium).
\end{proof}

\subsubsection{Stable price mechanism implementation}
At each round $t$, the system only needs to estimate the next round service value $\TM{S}{t+1}$, and thus a simpler ``one step lookahead'' oracle suffices, as opposed to the single-token case. Then, we have the following simple mechanism: 
\begin{enumerate}
\item  Users use (spend) $\tokenB$ tokens proportionally to $\tm{S}$ and buy new $\tokenB$ tokens proportionally to the next round service value $\TM{S}{t+1}$. 
    \item Validators' rewards consist of only tokens $\tokenB$ proportionally to the next round service value $\TM{S}{t+1}$.
    \item Validators sell all these $\tokenB$ tokens that they get as rewards, and keep all  $\tokenA$ tokens they had from the beginning. 
\end{enumerate}
It is worth pointing out that in our two-token model, users \emph{can} also buy tokens $\tokenA$, and validators \emph{can} sell them as well. However, in the equilibria induced by the proposed mechanisms, it is not convenient for either party to do so.

\begin{remark} 
We note that the equilibria described above are closely related to certain proposed implementations of IOTA \citep{saa2023iota}. As in our proposed reward scheme and equilibria,  rewards in IOTA system are uniquely provided in Mana tokens (= token $\tokenB$), while staking is conducted via IOTA tokens (= token $\tokenA$).  Furthermore, only Mana tokens can be traded to ``[...] individuals that in some sense are external to the system'' \citep[page~9]{saa2023iota}. In this context, our equilibria demonstrate that such designs, which forbid trading governance tokens, may still be sustainable. Specifically, Theorem~\ref{th:two-token-equilibria-exist} says that there exist equilibria where, even if trading governance tokens indirectly via some off-chain/secondary market became possible, players would not find it advantageous to do so.
\end{remark}

\section{Conclusions and research directions} 

We studied the problem of token monetary policy in blockchain systems viz-\`a-viz a set of desiderata that are crucial for the security and sustainability of blockchain systems.
Our analysis lead to our proposal of a new general tool
for tokenomics policy, quantitative  rewarding, which judiciously adjusts the rewards given to validators with the objective of achieving, at equilibrium, desirable paths for the token price, e.g., stable prices. In the light of this, 
we investigated how the long-term equilibria between users, validators and token flows are different for blockchains with a single token (used for transaction fees and staking) and two separate tokens. An important finding of this analysis, is that 
the two-token model affords additional flexibility that can handle a broader variation of service values at equilibrium. 
Specifically, on the issue of implementation of monetary policies satisfying all our desiderata, the results for single token tokenomics in Section~\ref{sec:single-token} highlight the need for some ``global'' information about future fluctuations in the service value, making  such policies difficult to realize (namely, in order to avoid buybacks and price fluctuations,  oracles that predict future shocks are needed to avoid reward explosion). In two-token systems, instead, the policy only needs to know the near future service value to control prices and enforce all our desiderata, a more tractable objective.  
In summary, in order to facilitate all our desiderata for a monetary policy, a suitable proxy for the service value $\tm{S}$ is essential, and for two-token systems the presence of such proxy can be put to use much more effectively via our proposed QR mechanism. 
Future work could delve  into assessing 
various possible proxies and their on-chain implementation 
as well as the 
incentive compatibility for users and validators 
that now will have the additional leverage of observing any of the oracle computations and data dependencies that may be taking  place publicly on-chain
and could be influenced by their actions.

\paragraph{Acknowledgments.} We are grateful to Christian Badertscher and Manvir Schneider for insightful comments and discussions on an earlier version of this work.

\bibliography{tokenomics}

\appendix
\newpage

\section{Postponed Proofs}
\subsection{Proofs of Section~\ref{sec:single-token}}\label{app:proof-steps}

\subsubsection{Proof of Lemma~\ref{le:prices-general}}
In the following, we consider a generic  symmetric equilibrium (Definition~\ref{def:weakly-symmetric-equilibrium}), 
\begin{align}
	(\tm{u_i},\tm{b_i}) = (\tm{u_U},\tm{b_U}) && \tm{s_j}=\tm{s_V} && \tm{\price}
\end{align}
for every user $i$ and every validator $j$, where the above strategies give the corresponding token holdings via \eqref{eq:token-holdings}. The proof follows the same steps as in \cite[Proof of Lemma~1]{decentralization_friction}. 
\begin{enumerate}
	\item \textbf{Validators' part.}
	Denote by $\VE_j^t$ the discounted utility of validator $j$ at equilibrium starting from step $t$.
	By the Principle of Optimality, 
	\begin{align}
		\VE_j^t = \max_{\tm{s_j}}[\tm{\price}\cdot \tm{s_{j}} - v] + \delta \VE_j^{t+1}
	\end{align}
	subject to the selling constraint \eqref{eq:val-sell-constraint}.
	We next show that we must have
	\begin{align}\label{eq:optimality-derivative-val}
		\frac{\partial \VE_j^t}{\partial \tm{s_{j}}} = 0  && \Longrightarrow && \tm{\price} = \delta \frac{\partial \VE_j^{t+1}}{\partial \TM{\token_{j}}{t+1}} 
	\end{align}
	The LHS  uses (crucially) that in the equilibrium (Definition~\ref{def:weakly-symmetric-equilibrium} and Remark~\ref{rem:positive-token-holding}) the selling constraint \eqref{eq:val-sell-constraint} is not binding (not strict), thus implying that the derivative must be zero at equilibrium.  The  RHS holds because 
	\begin{align}
		\frac{\partial \VE_j^t}{\partial \tm{s_j}} = \tm{\price} + \delta \frac{\partial \VE_j^{t+ 1}}{\partial \tm{s_j}} && \text{ and } &&  \frac{\partial \VE_j^{t+ 1}}{\partial \tm{s_j}} = \frac{\partial \TM{\token_{j}}{t+1}}{\partial \tm{s_j}}\frac{\partial \VE_j^{t+ 1}}{\partial \TM{\token_j}{t+1}}  \stackrel{\eqref{eq:token-holdings}}{=} (-1) \frac{\partial \VE_j^{t+ 1}}{\partial \TM{\token_j}{t+1}}
	\end{align}
	Next observe that, since $\tm{s_j} = \tm{\token_j} + \tm{R_j} - \TM{\token_j}{t+1} $, we also have
	\begin{align}
		\frac{\partial \VE_j^t}{\partial \tm{\token_j}} = \tm{\price} \cdot \left(1 + \tm{R}\cdot \frac{(n-1)\tm{\token_V}}{(\tm{\token_j} + (n-1)\tm{\token_V})^2} \cdot r' \left(\frac{\tm{\token_j}}{\tm{\token_j} + (n-1)\tm{\token_V}}\right)\right)
	\end{align}
	Plugging this  into \eqref{eq:optimality-derivative-val}, we obtain
	\begin{align}
		\label{eq:analysis-val-price_proof-appendix}
		\TM{\price}{t-1} = \delta \cdot \tm{\price} \cdot \left(1 + \frac{\tm{R}}{\tm{\token_V}}\cdot \frac{n-1}{n^2} \cdot r' (\frac{1}{n}) \right) 
	\end{align}
	
	\item \textbf{Users' part.} Denote by $\UE_i^t$ the discounted utility of user $i$ 
	at equilibrium starting from step $t$, 
	By the Principle of Optimality, 
	\begin{align}
		\UE_i^t = \max_{(\tm{u_i},\tm{b_i})}[\tm{S}\cdot g(n) \cdot s(\tm{u_i},\tm{u_{-i}}) - \tm{\price}\cdot \tm{b_{i}}] + \delta \UE_i^{t+1}
	\end{align}
	where $(\tm{u_i},\tm{b_i})$ satisfies the corresponding buying constraint \eqref{eq:usr-sell-constraint}, and $\tm{u_{-i}}$ denotes the use strategies of all other users. We next show that we must have
	\begin{align}\label{eq:optimality-derivative-usr}
		\frac{\partial \UE_i^t}{\partial \tm{b_{i}}} = 0  && \Longrightarrow && \tm{\price} = \delta \frac{\partial \UE_i^{t+1}}{\partial \TM{\token_{i}}{t+1}} 
	\end{align}
	The LHS  uses (crucially) that in the equilibrium (Definition~\ref{def:weakly-symmetric-equilibrium} and Remark~\ref{rem:positive-token-holding}) the buying constraint \eqref{eq:usr-sell-constraint} is not binding (not strict), thus implying that the derivative must be zero at equilibrium. 
	The  RHS holds because 
	\begin{align}
		\frac{\partial \UE_i^t}{\partial \tm{b_i}} = -\tm{\price} + \delta \frac{\partial \UE_i^{t+ 1}}{\partial \tm{b_i}} && \text{ and } &&  \frac{\partial \UE_i^{t+ 1}}{\partial \tm{b_i}} = \frac{\partial \TM{\token_{i}}{t+1}}{\partial \tm{b_i}}\frac{\partial \UE_i^{t+ 1}}{\partial \TM{\token_i}{t+1}}  \stackrel{\eqref{eq:token-holdings}}{=} (+1) \frac{\partial \UE_i^{t+ 1}}{\partial \TM{\token_i}{t+1}}
	\end{align}
	We next show that 
	\begin{align}\label{eq:analysis-usr-price-derivative-proof-appendix}
		\frac{\partial \UE_i^t}{\partial \tm{\token_i}} = 
		\begin{cases}
			\tm{S} \cdot g(n) \cdot \frac{(m-1)\tm{\token_U}}{(\token_i^t + (m-1)\tm{\token_U})^2} \cdot s'(\frac{\token_i^t}{\token_i^t + (m-1)\tm{\token_U}}) & \text{if at equilibrium } \tm{u_i} = \tm{\token_U}\\
			\tm{\price} & \text{if at equilibrium } \tm{u_i} < \tm{\token_U}
		\end{cases}
	\end{align}
	thus implying
	\begin{align}
		\label{eq:analysis-usr-price-proof-appendix}
		\TM{\price}{t-1} = \delta \cdot \begin{cases}
			\frac{\tm{S}}{\tm{\token_U}} \cdot g(n) \cdot \frac{m-1}{m^2}\cdot s'(\frac{1}{m}) & \text{if at equilibrium } \tm{u_i} = \tm{\token_U} \\
			\tm{\price} & \text{if at equilibrium } \tm{u_i} < \tm{\token_U} 
		\end{cases}
	\end{align}
	We distinguish the two cases in \eqref{eq:analysis-usr-price-derivative-proof-appendix}. If at equilibrium $\tm{u_i} = \tm{\token_U}$, then $\tm{b_{i}} = \TM{\token_i}{t+1} $ and 
	\begin{align}
		\UE_i^t = [\TM{S}{t} \cdot g(n) \cdot s(\tm{\token_i}, \tm{\token_{-iU}}) - \tm{\price}\cdot \TM{\token_i}{t+1}] + \delta \UE_i^{t+1}
	\end{align}
	where $\tm{\token_{-iU}}$ denotes the token holdings at equilibrium of all but user $i$ at step $t$. 
	The above equation implies that the corresponding derivative satisfies 
	\begin{align}
		\frac{\partial \UE_i^t}{\partial \tm{\token_i}} = \TM{S}{t} \cdot g(n) \cdot  \frac{m-1}{m^2}\cdot s'(\frac{1}{m})
	\end{align}
	If at equilibrium $\tm{u_i} < \tm{\token_U}$, then $\tm{b_i}= \TM{\token_i}{t+1} - \tm{\token_i}+ \tm{u_i}$ and 
	\begin{align}
		\UE_i^t = [\TM{S}{t} \cdot g(n) \cdot s(\TM{u_i}{t}, \TM{u_{-iU}}{t}) - \tm{\price}\cdot \tm{b_i}] + \delta \UE_i^{t+1}
	\end{align}
	thus implying that the corresponding derivative satisfies
	\begin{align}
		\frac{\partial \UE_i^t}{\partial \tm{\token_i}} = -\tm{\price}\cdot \frac{\partial \tm{b_i}}{\partial \tm{\token_i}} = -\tm{\price}\cdot  (-1) = \tm{\price} 
	\end{align}
	as in this case the optimal $\tm{u_i}<\tm{\token_U}= \tm{\token_i}$ is not affected by $\tm{\token_i}$.
\end{enumerate}

\subsubsection{Proof of Theorem~\ref{tH:prices_implies_generic}}
Condition \eqref{eq:price-growth} on the prices and Lemma~\ref{le:prices-general} imply $\delta\tm{\Rew}= \gamma$, and therefore $\tm{\Rew} = \gamma / \delta \neq 1$. The desired result follows from the definition of $\tm{\Rew}$ and $\tm{\Ser}$ and from  identity \eqref{eq:usr-val-connection} from Lemma~\ref{le:prices-general} (second part).  More in detail, we have
\begin{align}
	\tm{\Rew} \stackrel{\eqref{eq:val-price}}{=} 1 + \frac{\tm{R}}{n\tm{\token_V}}\cdot \frac{n-1}{n} \cdot r' (\frac{1}{n}) = 1 + \tm{\rewtostake} \cdot \kappa_R \ . 
\end{align}
Moreover 
\begin{align}
	\tm{\Rew} \stackrel{\eqref{eq:usr-val-connection}}{=} \frac{\tm{\Ser}}{\tm{\price}}   \stackrel{\eqref{eq:usr-price}}{=} \frac{\tm{S}}{\tm{\price} \cdot m\tm{\token_U}} \cdot g(n) \cdot \frac{m-1}{m}\cdot s'(\frac{1}{m}) = 
	\tm{\sertofee} \cdot \kappa_S \ . 
\end{align}
Since $\tm{\Rew} = \gamma / \delta$, the two identities above imply that   $\tm{\rewtostake} = \rewtostake$ and $\tm{\sertofee} = \sertofee$, for  constants $\rewtostake$ and $\sertofee$.
Hence, we have 
\begin{align}
	\TM{\price}{t-1}  \stackrel{\eqref{eq:val-price}}{=}  \tm{\price} \cdot \delta \cdot  \tm{\Rew} = & \tm{\price} \cdot \delta \cdot  (1 + \rewtostake \cdot \kappa_R ) \\
	= & \tm{\price} \cdot \delta \cdot  \tm{\sertofee} \cdot \kappa_S \ . 
\end{align}
We thus have proved that the first three conditions of Definition~\ref{def:symm-eq-general} hold. The last condition in Definition~\ref{def:symm-eq-general} on the token holdings follows from $\tm{\Rew} \neq 1$ and from the last part of Lemma~\ref{le:prices-general}  which states that  
$\tm{u_i} = \tm{\token_U}$ and $\tm{b_i}=\TM{\token_U}{t+1}$.

\subsubsection{Proof of Corollary~\ref{cor:stable-price-equilibria}}
Since we have stable prices,  $\tm{\price}=1$, Theorem~\ref{tH:prices_implies_generic} and Definition~\ref{def:symm-eq-general} imply that the token holdings must satisfy 

\begin{align*}
	\frac{\tm{S}}{m\tm{\token_U}} = \sertofee\ ,   
	&& 
	\frac{\tm{R}}{n\tm{\token_V}} = \rewtostake\ .  
\end{align*}
The corollary then follows from \eqref{eq:constants-stable-case}.

\subsubsection{Proof of Corollary~\ref{cor:reward-growth-single}} Simply observe that the minimal rewards satisfying \eqref{eq:th:gen-symm-equilibrium-no-buy-back} must satisfy that condition with equality, that is, 
\begin{align}
	\TM{R}{t+1} = \tm{R}\cdot(1+a)   - b \cdot \TM{S}{t+1} \cdot (\delta \Rew)^{t+1}
\end{align}
and, by induction on $t$, this yields $\tm{\Roptsingle}$ as in Corollary~\ref{cor:reward-growth-single}.

\subsection{Proofs of Section~\ref{sec:two-token-model}}\label{sec:two-token-analysis}

\subsubsection{Proof of Lemma~\ref{lem:prices-two-tokens}}
By adapting the analysis in \cite{decentralization_friction}, 
the  discounted utility at equilibrium of validator $j$ starting from step $t$ satisfies
\begin{align}\label{eq:val-utility-two-tokens}
	\tm{\VE_j} = \max_{\tm{s_j}}[\tm{\priceA}\cdot  \tm{\sellA{j}} + \tm{\priceB}\cdot \tm{\sellB{j}} - v] + \delta \TM{\VE_j}{t+1} 
\end{align}
where in a  symmetric equilibrium the selling constraints  \eqref{eq:val-sell-constraint-A} and \eqref{eq:val-sell-constraint-B} for validator $j$ become 
\begin{align}
	\label{eq:val-feasible-token-A}
	\TM{\tokenA_j}{t+1}  =  \tm{\tokenA_j} + \tm{\RRA} \frac{\tm{\tokenA_j}}{\tm{\tokenA_j} + (n-1)\tm{\tokenA_V}} - \tm{\sellA{j}} \geq 0 \\
	\label{eq:val-feasible-token-B}
	\TM{\tokenB_j}{t+1} =  \tm{\tokenB_j} + \tm{\RRB} \frac{\tm{\tokenA_j}}{\tm{\tokenA_j} + (n-1)\tm{\tokenA_V}} - \tm{\sellB{j}}\geq 0
\end{align} 
(Note that the two expressions above are \emph{not} symmetric as token $\tokenA$ is the only one used for staking -- the fraction term is identical.)
We next show that, under the optimal selling policy at equilibrium $\tm{s_j}$ we must have
\begin{align}\label{eq:optimality-derivative-two-tokens-A}
	\frac{\partial \tm{\VE_j}}{\partial \tm{\sellA{j}}} = 0  && \Longleftrightarrow && \tm{\priceA} = \delta \frac{\partial \TM{\VE_j}{t+1}}{\partial \TM{\tokenA_j}{t+1}} 
\end{align}
The LHS  uses (crucially) that in the equilibrium (Definition~\ref{def:weakly-symmetric-equilibrium} and Remark~\ref{rem:positive-token-holding}) the selling constraint \eqref{eq:val-feasible-token-A} is not binding (not strict), thus implying that the derivative must be zero.  The equivalence with the RHS is because 
\begin{align}
	\frac{\partial \tm{\VE_j}}{\partial \tm{\sellA{j}}} = \tm{\priceA} + \delta \frac{\partial \TM{\VE_j}{t+1}}{\partial \tm{\sellA{j}}} && \text{ and } &&  \frac{\partial \TM{\VE_j}{t+1}}{\partial \tm{\sellA{j}}} = \frac{\partial \TM{\tokenA_j}{t+1}}{\partial \tm{\sellA{j}}}\frac{\partial \TM{\VE_j}{t+1}}{\partial \TM{\tokenA_j}{t+1}}  = (-1) \frac{\partial \TM{\VE_j}{t+1}}{\partial \TM{\tokenA_j}{t+1}}
\end{align}
where $\frac{\partial \TM{\tokenA_j}{t+1}}{\partial \tm{\sellA{j}}}=-1$ uses (again) that the selling constraint \eqref{eq:val-feasible-token-A} is not binding (not strict). The exact same proof applies to the other token $\tokenB$, using the corresponding selling constraint \eqref{eq:val-feasible-token-B} to obtain
\begin{align}\label{eq:optimality-derivative-two-tokens-B}
	\frac{\partial \tm{\VE_j}}{\partial \tm{\sellB{j}}} = 0  && \Longleftrightarrow && \tm{\priceB} = \delta \frac{\partial \TM{\VE_j}{t+1}}{\partial \TM{\tokenB_{j}}{t+1}} 
\end{align}
Next observe that
\begin{align}
	\tm{\VE_j} = &  \tm{\priceA}  \tm{\sellA{j}} +  \tm{\priceB} \tm{\sellB{j}} - v   + \delta \VE_j^{* t+1} \\
	= &  \tm{\priceA}  \left[\tm{\tokenA_j} + \tm{\RRA} \frac{\tm{\tokenA_j}}{\tm{\tokenA_j} + (n-1)\tm{\tokenA_V}} - \TM{\tokenA_j}{t+1}\right] +  \\ 
	& \tm{\priceB} \left[\tm{\tokenB_j} + \tm{\RRB} \frac{\tm{\tokenA_j}}{\tm{\tokenA_j} + (n-1)\tm{\tokenA_V}} - \TM{\tokenB_j}{t+1} \right]  - v   + \delta \TM{\VE_j}{t+1} 
\end{align}
thus implying
\begin{align}
	\frac{\partial \tm{\VE_j}}{\partial \tm{\tokenA_j}} = & \tm{\priceA} \left[ 1 + \tm{\RRA} \frac{(n-1)\tm{\tokenA_V}}{(\tm{\tokenA_j} + (n-1)\tm{\tokenA_V})^2}\right] + \tm{\priceB} \left[\tm{\RRB} \frac{(n-1)\tm{\tokenA_V}}{(\tm{\tokenA_j} + (n-1)\tm{\tokenA_V})^2}\right] \\
	\frac{\partial \tm{\VE_j}}{\partial \tm{\tokenB_{j}}} = & \tm{\priceB} 
\end{align}
Plugging these into \eqref{eq:optimality-derivative-two-tokens-A} and \eqref{eq:optimality-derivative-two-tokens-B}, we obtain
\begin{align}
	\TM{\priceA}{t-1} =& \delta \tm{\priceA} \left[ 1 + \frac{\tm{\RRA}}{\tm{\tokenA}_V} \frac{n-1}{n^2}\right] + \delta \tm{\priceB} \left[ \frac{\tm{\RRB}}{\tm{\tokenA_V}} \frac{n-1}{n^2}\right]\\
	\TM{\priceB}{t-1} =& \delta \tm{\priceB}  
\end{align}
which proves Lemma~\ref{lem:prices-two-tokens}.

\subsubsection{Proof of Lemma~\ref{lem:prices-two-tokens-usr-side}}

The proof is identical to the one for the single token in Section~\ref{app:proof-steps}, users part. 
The equation for $\tokenB$ comes from the same analysis of the single token. As for token $\tokenA$, we observe that since token $\tokenA$ does not affect $\tm{S}_i$,  the  value of the service allocated to $i$, we have  $$\frac{\partial \tm{\UE_i}}{\partial \tm{\tokenA{_i}}} = \tm{\priceA}$$ for all $t$.  The  optimality conditions of strategy $\buyA{i}$ also yields $$\tm{\priceA}  =  \delta \frac{\partial \TM{\UE_i}{t+1}}{\partial \tm{\buyA{j}}}.$$ The latter quantity, by the previous equation with $t+1$ in place of $t$, equals $\delta \TM{\priceA}{t+1}$.

\subsubsection{Proof of Theorem~\ref{th:symm-eq-two-token}}
The first three items in the definition of generic symmetric equilibrium (Definition~\ref{def:symm-two-tokens}) follow from Lemma~\ref{lem:prices-two-tokens-usr-side},  from the definition of symmetric equilibrium requiring that users selling constraints of token $\tokenB$ are non-strict  (Definition~\ref{def:symm-two-tokens}), and from the assumption about the prices \eqref{eq:price-growth}. 

The last condition in Definition~\ref{def:symm-two-tokens} is a rewriting of the prices in Lemma~\ref{lem:prices-two-tokens} and of definition of symmetric equilibrium requiring that  validators selling constraints of token $\tokenA$ are non-strict (Definition~\ref{def:symm-two-tokens}).

\subsection{Price analysis for the proposed two-token equilibrium in Section~\ref{sec:stable-price-mechanism-two}}\label{app:proof-priceA-proposed-equilibrium}
We show Equation~\ref{eq:priceA-proposed-equilibrium}, which we reproduce here for convenience:
\begin{align}\label{eq:priceA-proposed-equilibrium-app}
	\tm{\priceA} & =  \frac{\TM{\priceA}{t-1}}{\delta} -  \TM{S}{t+1}\cdot  \frac{L }{\TM{\tokenA_V}{0}} \frac{n-1}{n^2}   \ . 
\end{align}
From the equation of the prices \eqref{eq:val-priceA-two-token}, we obtain 
\begin{align}
	\TM{\priceA}{t-1} & = \delta  \cdot \tm{\priceA} \cdot (1 +\tm{\inflationA}) +  \delta \cdot \tm{\priceB} \cdot \tm{\inflationB} \\
	& = \delta  \cdot \tm{\priceA}  +  \delta \cdot \tm{\priceB} \cdot \frac{\tm{\RRB}}{n\tm{\tokenA}_{V}} \frac{n-1}{n}  \\
	& = \delta  \cdot \tm{\priceA}  +  \delta \cdot \TM{S}{t+1}\cdot L \cdot \frac{1}{n\TM{\tokenA_V}{0}} \frac{n-1}{n}  
\end{align}
and by rearranging the terms we get \eqref{eq:priceA-proposed-equilibrium-app}, that is, \eqref{eq:priceA-proposed-equilibrium}.

\section{Comparison with equilibria and price paths in \cite{decentralization_friction}}\label{app:eq:simpler}
 In this section, we show that the class of equilibria we study here is more general than those  in \cite{decentralization_friction}    which focuses on equilibria having a \emph{stronger} symmetry condition. The latter work further assumes that the service value is known and under the control of the designer, a very strong assumption that is infeasible to justify in practice. 

\paragraph{Deterministic service value} In 
\cite{decentralization_friction}, the author considers a setting where $\tm{S}$ can be chosen by the system,  though the system incurs some quadratic time-dependent cost, and the equilibria maximizing  the profit of the system in this setting  correspond to some $\tm{S}$ satisfying our Assumption~\ref{asm:service-level} (hence, a special case of ours). 

\paragraph{Restricted class of equilibria} The condition that tokens holding must be strictly positive at all time steps is identical to the definition of equilibrium in \cite{decentralization_friction}. The latter work focuses on a class of  equilibria with a stronger symmetry condition. 
\begin{remark}[stronger version  of symmetric equilibria]\label{rem:stron-equilibrium}
The results in \cite{decentralization_friction} mostly focus on equilibria that have two additional conditions, namely,  validators  hold a constant fraction of all tokens (constant staking),
    \begin{align}\label{eq:stronger-equilibrium-tokens}
        \frac{n \tm{\token_V}}{m \tm{\token_U} + n \tm{\token_V}} = \rho >0
    \end{align}
    and the reward is also a fraction of all current tokens (constant inflation from validation)
    \begin{align}\label{eq:stronger-equilibrium-rewards}
        \frac{\tm{R}}{m \tm{\token_U} + n \tm{\token_V}} = I > 0 
    \end{align}
\end{remark}

For the above class of equilibria, the next result on the prices holds.

\begin{lemma}[Lemma~1 in \cite{decentralization_friction}]\label{le:prices-restricted}
   For the case of reward sharing and service schemes in \eqref{eq:simpler-schemes}, the following holds. 
    In any symmetric equilibrium (satisfying \eqref{eq:stronger-equilibrium-tokens} and \eqref{eq:stronger-equilibrium-rewards})
\begin{align}\label{eq:price-restricted}
    \TM{\price}{t-1}  = \tm{\price} \cdot \delta \cdot  \left(1 + \frac{I}{\rho}\cdot \frac{n-1}{n}\right) \ . 
\end{align}
\end{lemma}

Our results in Section~\ref{sec:result} generalize the price analysis above to a more general class of reward  and service sharing schemes, and \emph{without} imposing the constant staking and constant inflation from validation conditions in Remark~\ref{rem:stron-equilibrium} --note also that Assumptions~\ref{asm:service-level}  and \ref{asm:schemes} hold in  \cite{decentralization_friction} as well. Indeed, one of the distinguishing features of our ``quantitative rewarding'' mechanisms is to ``adjust'' rewards to ``absorb'' service level fluctuations and shocks, meaning that we cannot keep constant staking and inflation from validation. This is essentially due to the fact that stable prices require a particular (constant) $I$ and $\rho$ above, which then \emph{violate the nobuyback condition} in case of such shocks. 

\section{An example of an alternative rewards scheme}\label{sec:special rewards}
\newcommand{\equilibriumprofile}{\frac{\tm{\token_j}}{\sum_v \tm{\token_v}}}
\newcommand{\symmequilibriumprofile}{\frac{\tm{\token_j}}{\tm{\token_j} + (n-1) \tm{\token_V}}}
Consider the following ``soft cap'' function based on the sigmoid (logistic) function $\sigma(x)= \frac{\mathrm{e}^x}{1+\mathrm{e}^x}$, 
\begin{align}
	softcap(x,\tau,T) := 1 - \sigma(T\cdot(x - \tau)) =  \frac{1}{\mathrm{e}^{T\cdot(x - \tau)} + 1} 
\end{align}
Our reward function uses this soft cap, which intuitively aims at capping each validator's stake to a fraction $1/n$ of the total stake:
\begin{align}
	r(x) := 2x\cdot softcap(x,1/n,T) = \frac{2x}{\mathrm{e}^{T\cdot(x - \frac{1}{n})} + 1} 
\end{align}
In this case, 
\begin{align}
	r'(x) = -\dfrac{2\left(\left(Tx-1\right)\mathrm{e}^{T\cdot\left(x-\frac{1}{n}\right)}-1\right)}{\left(\mathrm{e}^{T\cdot\left(x-\frac{1}{n}\right)}+1\right)^2}\ ,  && r'(1/n) =1 - T/2n \ . 
\end{align}
Hence the condition $r'(1/n)>0$ required in Assumption~\ref{asm:schemes} is satisfied for all $n>T/2$, meaning that the designer can set $T$ in order to be satisfied for the number $n$ of validators.

\section{Benchmark (necessity of quantitative rewarding)}
In this section, we consider as a benchmark a \emph{fixed} rewards mechanism, that is, the case in which the rewards $\tm{R}$ are fixed in advance independently of the actual service level changes $\tm{S}$. We then study how the latter affect the price of the token over time, when the tight no buyback condition is satisfied. We show below (Corollary~\ref{cor:benchmark}) that prices are \emph{not} stable even in a simple restriction of the service levels $\tm{S}$, unless the rewards are set properly to match certain conditions. Moreover, when rewards are kept constant, prices fluctuate proportionally to the service level fluctuations.  

We stress that the connection is possible due to our no buybacks condition, in particular its \emph{tight} version in which the system neither buys nor sells tokens ever:
\begin{align}\label{eq:nobuyback-tight}
    m  \tm{b_U} = n \tm{s_V}\ .  
\end{align}
Note that, if $\tm{u_i}=\tm{\token_U}$, then we have $\tm{b_U}=\TM{\token_U}{t+1}$.
We also have 
\begin{align*}
	n\tm{s_V} & =   n \left(\tm{\token_V} - \TM{\token_V}{t+1} + \tm{R}\cdot r(1/n)\right)  \\ & = \frac{\tm{R}}{\rewtostake}  - \frac{\TM{R}{t+1}}{\rewtostake}  + \tm{R} \cdot n \cdot r(1/n) \ . 
\end{align*}
The following is thus a direct consequence of Lemma~\ref{le:prices-general}.

\begin{corollary}\label{cor:benchmark}
    Consider geometric rewards, that is, $\tm{R}=\alpha^t \cdot \TM{R}{0}$ for all $t$ and for some $\alpha>0$.  In any symmetric equilibrium, if $\tm{u_i}=\tm{\token_U}$ for all $t$, then   prices are proportional to the (next round) service level divided by the current rewards, that is, 
    \begin{align}
        \tm{\price} = 
        H\cdot \TM{S}{t+1}/\tm{R}   \ ,
    && H = \delta  \cdot\frac{g(n)}{\left(\frac{1-\alpha}{\rewtostake}  + n \cdot r(1/n)\right)}   \cdot \frac{m-1}{m}\cdot s'(\frac{1}{m}) \ . 
    \end{align}
    This implies that \begin{enumerate}
        \item Price are stable if and only if the service level follows a geometric growth with the same rate, that is, $\tm{S}=\alpha^t \cdot \TM{S}{0}$; 
        \item For constant rewards ($\tm{R}=\TM{R}{0}$), prices are proportional to the (next round) service level, that is, $\tm{\price} = 
        H\cdot \TM{S}{t+1}/\TM{R}{0} $. 
    \end{enumerate}, 
\end{corollary}

\begin{proof} From the second part of Lemma~\ref{le:prices-general}, 
    \begin{align}
        \tm{\price} = \delta \cdot 
        \TM{\Ser}{t+1}  \ ,
    && \TM{\Ser}{t+1} = \frac{\TM{S}{t+1}}{m\TM{\token_U}{t+1}} \cdot g(n) \cdot \frac{m-1}{m}\cdot s'(\frac{1}{m}) \ . 
    \end{align}
    where 
    \begin{align}
       \TM{\Ser}{t+1} = \frac{\TM{S}{t+1}}{m\TM{\token_U}{t+1}} \cdot g(n) \cdot \frac{m-1}{m}\cdot s'(\frac{1}{m}) =  \frac{\TM{S}{t+1}}{n\tm{s_V}} \cdot g(n) \cdot \frac{m-1}{m}\cdot s'(\frac{1}{m})\ . 
    \end{align}
    To complete the proof, we observe that 
    \begin{align*}
	n\tm{s_V} & =   n \left(\tm{\token_V} - \TM{\token_V}{t+1} + \tm{R}\cdot r(1/n)\right)  \\ & = \frac{\tm{R}}{\rewtostake}  - \frac{\TM{R}{t+1}}{\rewtostake}  + \tm{R} \cdot n \cdot r(1/n) \\
 & = 
 \frac{\alpha^t \TM{R}{0}}{\rewtostake}  - \frac{\alpha^{t+1} \TM{R}{0}}{\rewtostake}  + \alpha^t \TM{R}{0} \cdot n \cdot r(1/n) \\
 & = 
 \alpha^t \TM{R}{0} \left(\frac{1-\alpha}{\rewtostake}  + n \cdot r(1/n)\right)\ .  
\end{align*}
\end{proof}

Note that for generic rewards, $\tm{R}$, the above constant $H$ will actually depend on $t$. Stable prices then boil down to ensure  that $\tm{H}\cdot \TM{S}{t+1}/\tm{R}$ is constant. 

\end{document}